\documentclass[
  aps,
  prx,
  twocolumn,
  superscriptaddress,
  nofootinbib,
  longbibliography
]{revtex4-2}

\usepackage{amsmath,amsthm,amssymb,amsfonts,bm,mathtools}
\usepackage{graphicx}
\usepackage{xcolor}
\usepackage{hyperref}
\usepackage[noend,noline,ruled,linesnumbered]{algorithm2e}

\hypersetup{colorlinks=true,linkcolor=blue,citecolor=blue,urlcolor=blue}
\DeclareMathOperator{\supp}{supp}
\DeclareMathOperator{\Tr}{Tr}

\DeclareMathOperator{\cscop}{csc}
\newcommand{\cB}{\mathcal B}

\newcommand{\cJ}{\mathcal J}
\newcommand{\cL}{\mathcal L}
\newcommand{\cP}{\mathcal P}

\newcommand{\cD}{\mathcal D}

\newcommand{\Id}{\mathsf I}

\newcommand{\E}{\mathbb E}
\newcommand{\R}{\mathbb R}

\newcommand{\cE}{\mathcal E}
\newcommand{\cF}{\mathcal F}

\newcommand{\cA}{\mathcal A}
\newcommand{\cM}{\mathcal M}
\newcommand{\norm}[1]{\left\lVert #1\right\rVert}
\newcommand{\abs}[1]{\left\lvert #1\right\rvert}
\newtheorem{lemma}{Lemma}
\newtheorem{corollary}{Corollary}
\newtheorem{proposition}{Proposition}
\newtheorem{problem}{Problem}

\newtheorem{remark}{Remark}
\newtheorem{theorem}{Theorem}

\newcommand{\pku}{Center on Frontiers of Computing Studies, School of Computer Science, Peking University, Beijing 100871, China}
\newcommand{\bnu}{School of Artificial Intelligence, Beijing Normal University, Beijing 100875, China}

\begin{document}

\title{Efficient Lindbladian Learning from Constant-Time Pauli Responses}

\author{Jiaxing Song}
\email{pkujiaxing@stu.pku.edu.cn}
\affiliation{\pku}
\author{Yukun Zhang}
\email{yukunzhang@stu.pku.edu.cn}
\affiliation{\pku}
\author{Xiao Yuan}
\email{xiaoyuan@pku.edu.cn}
\affiliation{\pku}
\author{Yusen Wu}
\email{yusen.wu@bnu.edu.cn}
\affiliation{\bnu}

\begin{abstract}
Learning the generator of an open many-body system is more challenging than Hamiltonian learning: local responses, which can directly reveal coherent interaction terms in closed-system dynamics, may also contain dissipative contributions in open-system dynamics. In this paper, we address this challenge by developing an efficient Lindbladian learning framework for a known local candidate generator dictionary with bounded dissipative support and either bounded dual-interaction-graph degree or bounded unweighted local strength. The framework resolves the coherent-dissipative ambiguity by treating local Pauli responses as a linear system over both types of generator terms. Inverting this response system separates their contributions and makes the individual Lindbladian coefficients accessible from local response data in a fixed short-time window. Within this framework, we develop two efficient learning algorithms: Chebyshev--Lobatto response interpolation, which uses logarithmically many short evolution times and has a post-mean cost linear in \(M\), with the stated dependence on \(\epsilon\), and Single-time projected response contraction, which uses a single fixed evolution time and globally inverts a truncated response function. Both procedures estimate \(M\) candidate coefficients to entrywise accuracy \(\epsilon\) using \(\widetilde{\mathcal{O}}(M/\epsilon^2)\) sample and classical post-processing complexity. Our theoretical results establish local response inversion as a scalable paradigm for learning, calibrating, and diagnosing complex quantum systems from experimentally accessible short-time data.
\end{abstract}

\maketitle

\begin{figure*}[t!]
    \centering
    \includegraphics[width=\textwidth]{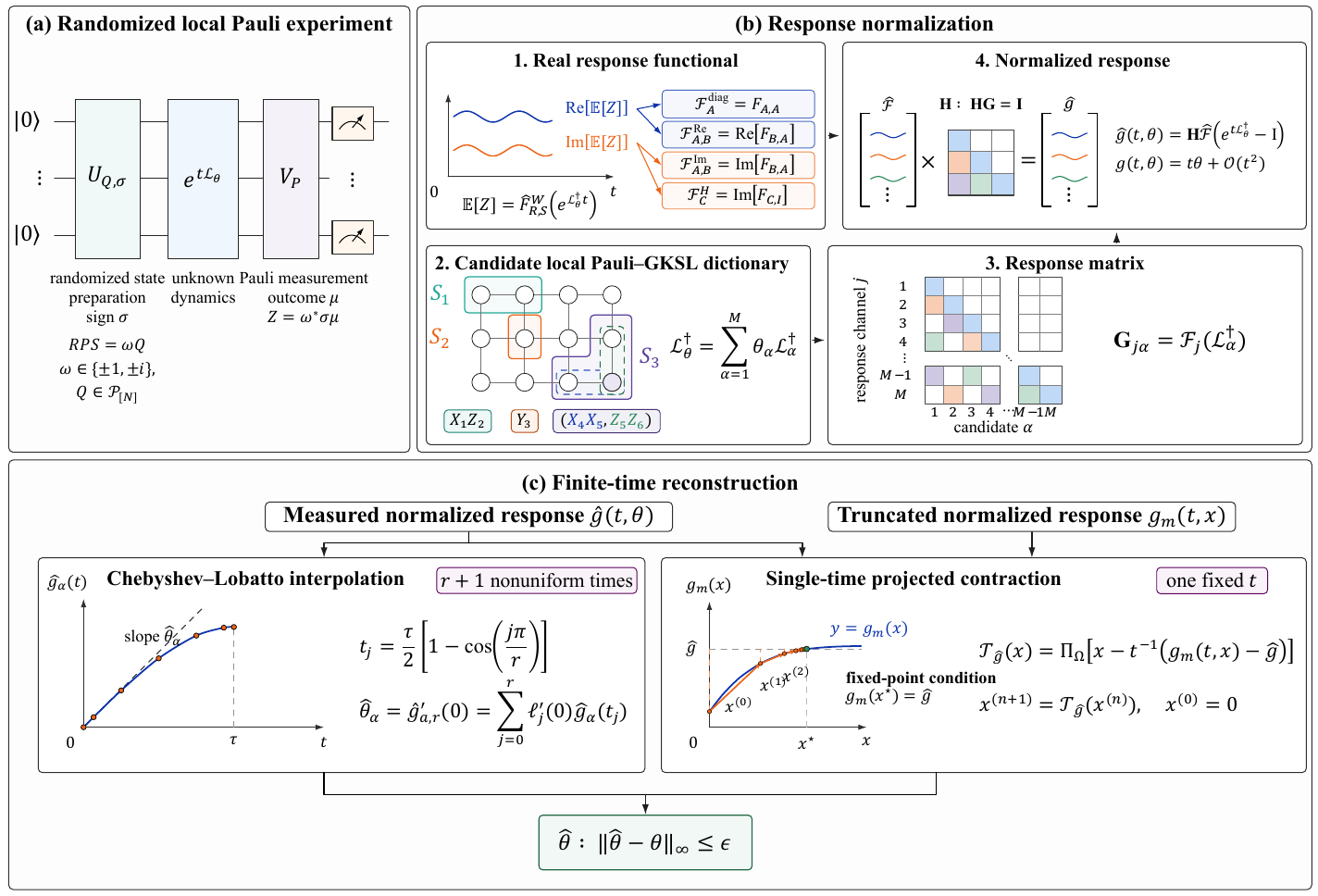}
    \caption{
Overview of the Lindbladian learning framework. (a) Randomized product-Pauli experiments estimate finite-time raw Pauli response
functionals. (b) Response normalization removes the intrinsic mixing between coherent and dissipative generator terms. The measured real response functionals are constructed from raw responses, while the candidate Pauli--GKSL dictionary defines the generator coordinates. At first order in the evolution time, the response matrix
\(\mathbf G\) maps Lindbladian coefficients to the slopes of these raw
responses at \(t=0\). Applying
\(\mathbf H=\mathbf G^{-1}\) removes coherent--dissipative and
common-extension mixing and produces normalized responses whose derivatives
at the origin equal individual generator coefficients. 
(c) Figure~\ref{fig:learning-workflow} summarizes the three logically distinct
stages of the framework: experimental estimation of raw Pauli responses,
dictionary-dependent response normalization through
\(\mathbf H=\mathbf G^{-1}\), and finite-time coefficient reconstruction by
either Chebyshev--Lobatto differentiation or single-time projected inversion.
}
    \label{fig:learning-workflow}
\end{figure*}

\section{Introduction}
Quantum computers are entering regimes in which their behaviors cannot be
characterized and verified by classical computation. Coherent control over large numbers of physical qubits has now been demonstrated across superconducting circuits, trapped ions, and neutral-atom arrays. As these platforms continue to scale, progress in quantum error correction, error mitigation, and device calibration increasingly depends on precise and scalable tools for characterizing many-body quantum dynamics~\citep{arute2019quantum,morvan2024phase,zhong2020quantum,smith2016many,evered2023high,bravyi2024high,acharya2024quantum,kimScalableErrorMitigation2023,o2023purification,kimEvidenceUtilityQuantum2023}. Standard quantum process tomography gives a complete description of a channel, but it requires a number of experimental settings and classical parameters that grows exponentially with the number of qubits~\citep{banaszekFocusQuantumTomography2013,blume-kohoutOptimalReliableEstimation2010,eisertQuantumCertificationBenchmarking2020,grossQuantumStateTomography2010,hradilQuantumstateEstimation1997,maurodarianoQuantumTomography2003,chuangPrescriptionExperimentalDetermination1997,darianoQuantumTomographyMeasuring2001,mohseniQuantumprocessTomographyResource2008}. Scalable quantum process learning must therefore avoid reconstructing the full channel where the underlying generator has a clear structure.

Hamiltonian learning provides a structured alternative to full process tomography when the generator admits a compact parametrization~\citep{gebhartLearningQuantumSystems2023,haahLearningQuantumHamiltonians2024}. Suppose the Hamiltonian governing the quantum dynamics is given by \(H=\sum_a\theta_a C_a\), and the goal is to learn the real-valued coefficients \(\theta_a\) given access to \(e^{-iHt}\) at tunable evolution times \(t\). A common approach in related works~\citep{HTFS2023,BakshiLiuMoitraTang2024,HuEtAl2025,PradenneCotlerHuang2026} is to construct local response functions \(\operatorname{Tr}[\rho P(t)]\) by carefully choosing the input state \(\rho\) and the local observable \(P\), and thereby recover the target coefficient vector \(\theta=(\theta_1,\theta_2,\ldots)^{\mathsf T}\). For each candidate term \(C_a\), one may choose an input state \(\rho_a\) and a local Pauli observable \(P_a\) such that \(\operatorname{Tr}[\rho_a P_a]=0\) and \(\{C_a,P_a\}=0\), yielding the local response function \(\operatorname{Tr}[\rho_a P_a(t)]=t\theta_a+\mathcal{O}(t^2)\). Locality then allows the higher-order corrections to be controlled using cluster-expansion bounds, so the coefficient can be estimated as \(\theta_a\approx \operatorname{Tr}[\rho_a P_a(t)]/t\). For Lindbladian dynamics, however, this coefficient-by-coefficient
identification generally fails. The term linear in \(t\) can receive
contributions from both the Hamiltonian and dissipative parts of the generator.
Moreover, even dissipative terms with different supports can produce the same
local response. Consequently, a local response generally reveals a linear
combination of Lindbladian coefficients rather than a single coefficient.

The preceding discussion suggests that Pauli responses should be analyzed
jointly rather than one at a time. We index the measured responses by \(j\) and
the candidate Lindbladian terms by \(\beta\), and define \(\mathbf G_{j\beta}\) as the
contribution of the \(\beta\)-th candidate term to the part of the \(j\)-th
response that is linear in \(t\). Let
\(\theta=(\theta_1,\ldots,\theta_M)^{\mathsf T}\)
collect all unknown coherent and dissipative coefficients. The measured vector of Pauli
responses from experiments then may take the form
\(t\mathbf G\theta+\mathcal{O}(t^2)\). In this way, both the mixing
between coherent and dissipative terms and the mixing among dissipative terms
are encoded in a known response matrix \(\mathbf G\). However, a response defined on a region \(W\) may also contain contributions from
dissipative terms supported on a larger region. This occurs when the larger
term extends both Pauli labels outside \(W\) by the same Pauli string. In
contrast, a term supported only on \(W\) cannot contribute to a response whose
Pauli labels act nontrivially outside \(W\). 
The mixing is therefore one-directional with respect to support inclusion, so
ordering the dissipative terms by support makes the corresponding block of
\(\mathbf G\) triangular.The
resulting square response matrix is invertible, with
\(\mathbf H=\mathbf G^{-1}\). Under the uniform bounded-overlap conditions
used below, the rows of \(\mathbf H\) contain only a constant number of nonzero
entries and have norms bounded independently of the system size. 
Multiplying the raw responses by \(\mathbf H\) removes both sources of mixing
and defines the normalized responses \(g_\alpha(t,\bm\theta)\). By
construction, the \(\alpha\)-th response satisfies
\(g_\alpha(t,\bm\theta)=t\theta_\alpha+\mathcal O(t^2)\). Thus, the term
linear in \(t\) in each normalized response depends on exactly one
Lindbladian coefficient. This resolves the algebraic identifiability
problem. The remaining problem is analytic: how to recover
\(\theta_\alpha\) from noisy finite-time values of
\(g_\alpha(t,\bm\theta)\).

We develop two complementary methods for extracting the Lindbladian
coefficients from noisy finite-time responses. Chebyshev--Lobatto response interpolation
uses \(\mathcal{O}(\log(1/\epsilon))\) short evolution times and combines the
resulting data through polynomial interpolation to estimate the term linear in
\(t\). Given the empirical response estimates, its classical reconstruction cost is
\(\mathcal O(Mr+r^2)\), where
\(r=\mathcal O(\log(1/\epsilon))\). In contrast, single-time projected response contraction uses data from a
single fixed evolution time and directly inverts the nonlinear finite-time
response map, trading fewer experimental time settings for a classical
reconstruction cost of
\(\widetilde{\mathcal O}(M/\epsilon)\) in the symbolic-Pauli model.
Despite these different tradeoffs, both methods estimate all \(M\) candidate
coefficients to entrywise accuracy \(\epsilon\) using
\(\widetilde{\mathcal O}(M/\epsilon^2)\) experimental shots. Including the
aggregation of individual measurement outcomes, the overall classical
response-processing cost is also
\(\widetilde{\mathcal O}(M/\epsilon^2)\). Taken together, our results extend the local-response paradigm underlying scalable Hamiltonian learning to Markovian open-system dynamics. They provide
a systematic framework for learning, calibrating, and diagnosing interacting
dissipative many-body systems using experimentally accessible Pauli
measurements.

Figure~\ref{fig:learning-workflow} summarizes the complete learning
pipeline. The procedure consists of three stages. First, randomized local
Pauli experiments provide unbiased estimators of experimentally accessible
raw responses. Second, these responses are transformed through a
dictionary-dependent response matrix inversion, which removes the mixing
between coherent and dissipative Lindbladian terms and constructs normalized
responses whose linear coefficients correspond directly to individual
generator parameters. Finally, the remaining nonlinear finite-time dependence
is resolved by either Chebyshev--Lobatto interpolation or single-time projected
response contraction.

\section{Learning Algorithm Outline}
We consider time-independent Markovian dynamics on \(N\) qubits, with the underlying Lindbladian operator given by
\begin{equation}
 \cL_{\bm\theta}^\dagger
 =\sum_{\alpha=1}^{M}\theta_\alpha\cL_\alpha^\dagger,
 \qquad
 \bm\theta\in\Theta_{\mathrm{GKSL}}\subseteq[-1,1]^M .
 \label{eq:model}
\end{equation}
We assume a known linearly independent, nonredundant real Pauli--GKSL
candidate dictionary $\mathcal A=\{\mathcal L_\alpha\}_{\alpha=1}^M$,
with the identity excluded from the dissipative Pauli basis as in
Appendix~\ref{app:param}.  We require the learning algorithm to follow a
simple ``prepare--evolve--measure'' paradigm using only the \(N\) system qubits, without ancillary qubits, and focus on recovering the unknown
parameters \(\bm\theta\) once the dictionary is known.

We assume that every dissipative candidate direction has support size at most
a constant \(S_D\), and write
\(s:=\max_{\alpha\in\mathcal A}
|\operatorname{supp}(\mathcal L_\alpha^\dagger)|\)
for the maximum support size over the full dictionary. The dual interaction
graph has one vertex for each candidate direction and joins two distinct
vertices exactly when their supports overlap; let \(d\) denote its
maximum degree.  We assume either bounded dual-interaction-graph degree or,
alternatively, bounded unweighted local strength, as defined in
Appendix~\ref{app:taylor}.  The latter is an alternative sufficient locality
characterization.  Hamiltonian candidates with support larger than \(S_D\)
are handled by the direct responses of Appendix~\ref{app:responses}. 

\begin{problem}[Lindbladian learning]
\label{problem}
    Given such a nonredundant candidate dictionary $\mathcal A$ for a
    Lindbladian operator \(\mathcal L_{\bm\theta}\)
    (Eq.~\eqref{eq:model}), the task is to estimate the intrinsic coupling
    coefficients \(\bm\theta\) within $\epsilon$ additive error, given access
    to the Lindbladian evolution \(e^{\mathcal{L}_{\bm\theta}t}\) for tunable
    evolution times \(t=\mathcal{O}(1)\).
\end{problem}

The central algebraic obstruction is that the coefficient of \(t\) in a single 
Pauli response need not equal a single generator coordinate. Coherent and 
dissipative candidates can contribute to the same response, and dissipative 
candidates related by a common external Pauli extension can also mix. We 
therefore introduce raw response functionals that expose this mixing as a known 
linear system. For a qubit subset \(W\subseteq[N]\), let \(\cP_W\) denote the 
phase-free
Pauli strings supported inside \(W\), tensored with identities outside \(W\)
when viewed as \(N\)-qubit operators. For \(R,S\in\cP_W\) and a
linear map $\Phi$ acting on $N$ qubits, define the \emph{raw Pauli response}
\begin{equation}
 F_{R,S}^{W}(\Phi)
 =\frac{1}{4^{|W|}}
 \sum_{P\in\cP_W}
 \frac{1}{2^N}\Tr\!\left[(RPS)^\dagger\Phi(P)\right].
 \label{eq:raw-response}
\end{equation}
Suppose the linear map $\Phi$ has a left--right Pauli expansion
\begin{equation}
 \Phi(O)=\sum_{\widetilde R,\widetilde S\in\cP_{[N]}}
 a_{\widetilde R,\widetilde S}\,
 \widetilde R O\widetilde S,
\end{equation}
then Pauli orthogonality gives
\begin{equation}
 F_{R,S}^{W}(\Phi)
 =\sum_{T\in\cP_{W^c}}
 a_{R\otimes T,S\otimes T}.
 \label{eq:common-extension}
\end{equation}

Equation~\eqref{eq:common-extension} shows that a response on \(W\) contains the coefficient associated with \(W\) together with contributions from its larger extensions. This suggests learning the coefficients jointly rather than extracting them one at a time. A candidate with maximal support has no larger extension in the dictionary, so its associated response has no larger-support contamination. Once the maximal-support coefficients are determined, their contributions can be subtracted from responses on smaller supports, and the procedure can be continued recursively. Ordering the candidates from larger to smaller support expresses this elimination as a triangular linear system. To implement the entire elimination as a single matrix inversion, we now construct a real response vector with one entry for each independent real generator coefficient.
Let \(M\) be the number of real candidate coordinates. The raw responses
\(F_{R,S}^W\) are generally complex and occur in conjugate pairs, whereas the
Lindbladian is parameterized by \(M\) independent real coefficients. We
therefore form one real response functional \(\mathcal F_j\) for each real
coordinate. For a Hamiltonian coordinate associated with \(C\), we use
\(
    \mathcal F_j
    =
    (
        F_{C,I}^{\operatorname{supp}(C)}
        -
        F_{I,C}^{\operatorname{supp}(C)}
    )/(2i)
\)
when \(|\operatorname{supp}(C)|\le S_D\), and the direct response
\(\mathcal R_C\) of Appendix~\ref{app:responses} when
\(|\operatorname{supp}(C)|>S_D\). For a diagonal dissipative coordinate
associated with \(A\), we take \(\mathcal F_j=F_{A,A}^W\); and for the real
and imaginary coordinates of an off-diagonal Kossakowski pair \(A<B\), we take
\(
\mathcal F_j=\frac12\left(F_{B,A}^W+F_{A,B}^W\right),
\qquad
\mathcal F_{j'}=-\frac{i}{2}\left(F_{B,A}^W-F_{A,B}^W\right),
\)
respectively. These combinations place the measured responses in one-to-one correspondence with the real Lindbladian coordinates. They do not remove the larger-support contributions; instead, they provide the \(M\) real equations in which those contributions can be represented by the response matrix \(\mathbf G\) and removed by inversion.

The resulting \(M\) coupled equations are summarized by the real response matrix
\begin{equation}
 \mathbf G_{j\beta}=\cF_j(\cL_\beta^\dagger).
 \label{eq:Gdef}
\end{equation}
Writing \(\boldsymbol{\mathcal F}(\cdot)=\{\mathcal F_j(\cdot)\}_{j=1}^M\), the linearity of the generator gives $\boldsymbol{\mathcal F}
\!\left(\mathcal L_{\boldsymbol\theta}^\dagger\right)
=
\mathbf G\boldsymbol\theta$. The diagonal entries of \(\mathbf G\) describe the intended response–coordinate pairings, while its off-diagonal entries record contributions from other candidate directions. Under the support ordering described above, the dissipative block is triangular with nonzero diagonal. The Hamiltonian terms form an identity block, with any dissipative contributions appearing as known off-diagonal entries. Hence the full response matrix is invertible. We define
\begin{equation}
 \mathbf{H}=\mathbf{G}^{-1},
 \qquad
 \mathsf C_\alpha(\cdot)
 =\sum_j\mathbf{H}_{\alpha j}\cF_j(\cdot),
 \label{eq:normalized-response}
\end{equation}
which gives rise to $ \mathsf C_\alpha(\cL_\beta^\dagger)=\delta_{\alpha\beta}$. Thus, each row of \(\mathbf H\) specifies the linear combination of raw responses that subtracts the contributions from the other candidate directions and isolates the coefficient \(\theta_\alpha\). This resolves the generator-level mixing, and the remaining task is to extract these coefficients from responses measured at nonzero evolution times.

Finally, we define the finite-time response function by
\begin{equation}
 g_\alpha(t,\bm\theta)
 :=\mathsf C_\alpha\!\left(e^{t\cL_{\bm\theta}^\dagger}-\Id\right).
 \label{eq:gdef}
\end{equation}
For $\Lambda t<1$, this response admits a controlled multivariate polynomial
approximation in the full vector \(\bm\theta\), whose linear term is
\(t\theta_\alpha\). Here \(\Lambda=2(d+1)\) in the bounded-degree
setting and is the corresponding unified locality scale under the alternative
local-strength assumption. Rigorous bounds are given in
Appendix~\ref{app:response-taylor}. The learning procedure therefore has three stages:
\begin{enumerate}
    \item compute \(\mathbf G\) and \(\mathbf H=\mathbf G^{-1}\) once from
    the known candidate dictionary;
    \item estimate the raw response means and apply \(\mathbf H\) to form the
    normalized finite-time responses \(g_\alpha(t,\bm\theta)\);
    \item recover \(\bm\theta\) using either Chebyshev--Lobatto endpoint
    differentiation or single-time projected response inversion.
\end{enumerate}

\section{Experimental Response Estimation}
Each real raw response admits a bounded unbiased estimator obtained from a
randomized product-Pauli preparation, short-time evolution, and Pauli
measurement. Concretely, sample \(P\in\mathcal P_W\) uniformly and write
\(RPS=\omega(P)Q(P)\), where
\(\omega(P)\in\{\pm1,\pm i\}\). When \(Q(P)\neq I\), the required average
input state \((I+\sigma Q(P))/2^N\), with
\(\sigma\in\{\pm1\}\), is implemented by sampling product Pauli eigenstates
from an appropriate ensemble. The \(Q(P)=I\) branch is implemented by a product-state ensemble with
maximally mixed average. After evolution under
\(e^{t\mathcal L_{\bm\theta}}\), one measures the Pauli observable \(P\) and
records the corresponding phase-weighted outcome \(Z\).

The complete sampling rule is given in the Appendix \ref{app:measurement}. In either branch,
\(|Z|\le1\) and
\(\mathbb E[Z]
=F_{R,S}^{W}(e^{t\mathcal L_{\bm\theta}^\dagger})\). Taking the appropriate
real or imaginary component therefore gives a bounded unbiased estimator of
each real raw response \(\mathcal F_j\). Applying \(\mathbf H\) to the
resulting empirical means produces estimates of the normalized responses.

We define
\(L_C:=\max_\alpha\sum_j|\mathbf H_{\alpha j}|\), which controls the
amplification of statistical errors when the raw responses are combined using
\(\mathbf H\). Under the uniform bounded-overlap conditions
\(S_D,\Lambda=\mathcal O(1)\),
Theorem~\ref{thm:uniform-response-inverse} gives
\(L_C=\mathcal O(1)\). Since \(|Z|\le1\), estimating all \(M\) normalized
response coordinates at a fixed evolution time to coordinatewise error at
most \(\varepsilon_g\), with probability at least \(1-\delta\), requires
\(\mathcal O(M\log(M/\delta)/\varepsilon_g^2)\) independent experimental
shots.
\section{Main Results}

In the informal results below, we work in the uniform bounded-overlap regime
\(S_D,\Lambda,L_C=\mathcal O(1)\). Throughout, classical post-processing
includes aggregating the measurement outcomes into empirical response
estimates and reconstructing \(\widehat{\bm\theta}\) from those estimates. Reconstruction costs are evaluated in the
unit-cost symbolic-Pauli model of
Appendix~\ref{subsec:classical-coefficient-construction}.

\emph{Chebyshev method.---}Since
$g_\alpha(t,\bm\theta)=t\theta_\alpha+\mathcal{O}(t^2)$, the quotient
$g_\alpha(t,\bm\theta)/t$ has $\theta_\alpha$ as its leading term. However,
using a single evolution time leads to a bias--variance tradeoff: a longer
time increases the higher-order bias, whereas a shorter time amplifies the
statistical error by $1/t$. To avoid choosing a single compromise time, we
evaluate the response at the shifted Chebyshev--Lobatto nodes
\begin{equation}
 t_j=\frac{\tau}{2}\left(1-\cos\frac{j\pi}{r}\right),
 \qquad j=0,\ldots,r,
 \qquad \tau=\frac{1}{2\Lambda}.
 \label{eq:nodes-main}
\end{equation}
Let $\widehat g_{\alpha,r}(t)$ be the degree-$r$ interpolant of the measured
response values. We estimate the coefficient as
$\widehat\theta_\alpha:=\widehat g_{\alpha,r}'(0)$.

\begin{theorem}[Informal]
\label{thm:cheb-main}
Consider the \(N\)-qubit Lindbladian learning problem defined in
Problem~\ref{problem}. For \(0<\epsilon,\delta<1\), choose the interpolation
parameters according to the prescription in the Appendix \ref{app:chebyshev}. The resulting
 degree of interpolation satisfies
\(r=\mathcal{O}\!\left(\log(L_C\Lambda/\epsilon)\right)\), corresponding to
the \(r\) nonzero Chebyshev--Lobatto evolution times. The method uses
\(\widetilde{\mathcal{O}}(M/\epsilon^2)\) experimental shots and outputs an
estimate \(\widehat{\bm\theta}\) satisfying
\(\|\widehat{\bm\theta}-\bm\theta\|_\infty\le\epsilon\) with probability at
least \(1-\delta\). Given the empirical response estimates, reconstructing all
coefficients requires \(\mathcal{O}(Mr+r^2)\) classical operations. Including
the aggregation of measurement outcomes, the overall classical
post-processing complexity is
\(\widetilde{\mathcal{O}}(M/\epsilon^2)\).
\end{theorem}

Once the empirical means have been collected, the Chebyshev procedure estimates each coefficient using a precomputed weighted sum of the responses at the sampled times and requires no iterative solver. It does, however, require a logarithmic number of prescribed evolution times fixed by the interpolation rule. This motivates a complementary method that uses a single, well-calibrated evolution time.

\emph{Single-time projected response contraction.---}
The key idea is to recast learning as solving the finite-time response
equation
\(g(t,\bm\theta)=\widehat{\bm g}\), where
\(\widehat{\bm g}\) is obtained from measurements at a single fixed
evolution time.  We take the canonical evolution time
\begin{equation}
    t=t_\star
    :=
    \frac{1}{4\Lambda(1+6L_C\Lambda)}.
    \label{eq:main-canonical-projected-time}
\end{equation}
Under the uniform bounded-overlap assumptions,
\(t_\star=\Theta(\Lambda^{-2})\).  Although the exact response function cannot
generally be computed classically, for any trial vector
\(x\in[-1,1]^M\), the known local generator dictionary allows us to evaluate
a degree-\(m\) approximation \(g_m(t,x)\).  This evaluation uses the truncated local Taylor expansion and sparse symbolic
enumeration developed in the Appendix.\ref{app:single-time}, and its raw components are combined through
\(\mathbf H\) to produce an approximation to the normalized response.

This classical forward map allows us to solve the response equation
through the projected iteration starting from \(x^{(0)}=0\):
\begin{align}
x^{(n+1)}
=
\Pi_{[-1,1]^M}\!\left[
x^{(n)}
-t^{-1}\bigl(
g_m(t,x^{(n)})-\widehat{\bm g}
\bigr)
\right].
\end{align}
The key reason for the efficiency of the iteration process is that the normalized response map is close to the linear map \(x\mapsto tx\).
More precisely, either locality assumption implies that, at the canonical
time \(t=t_\star\), for any \(x,z\in[-1,1]^M\),
\begin{align}
\left\|
t^{-1}\!\left[
g_m(t,x)-g_m(t,z)
\right]
-
(x-z)
\right\|_\infty
\le
\frac{1}{2}\|x-z\|_\infty,
\end{align}
uniformly over the parameter domain.  Thus the scaled response mismatch
\(t^{-1}(g_m(t,x^{(n)})-\widehat{\bm g})\) estimates the current
parameter error to a relative accuracy $1/2$.  Subtracting it in the
projected update removes the leading error, and, in the absence of
truncation and statistical errors, gives
\begin{align}
\|x^{(n+1)}-\bm\theta\|_\infty
\le
\frac{1}{2}\|x^{(n)}-\bm\theta\|_\infty .
\end{align}
Consequently,
the parameter error decreases geometrically and only $\mathcal{O}(\log (1/\epsilon))$ iterations are required.

\begin{theorem}[Informal]
\label{thm:projected-main}
Consider the \(N\)-qubit Lindbladian learning problem defined in
Problem~\ref{problem}. Use the canonical evolution time \(t=t_\star\) defined
in Eq.~\eqref{eq:main-canonical-projected-time}, with the remaining
algorithmic parameters chosen according to the prescription in the Appendix \ref{app:single-time}.
The single-time projected response method uses
\(\widetilde{\mathcal{O}}(M/\epsilon^2)\) experimental shots and outputs an
estimate \(\widehat{\bm\theta}\) satisfying
\(\|\widehat{\bm\theta}-\bm\theta\|_\infty\le\epsilon\) with probability at
least \(1-\delta\). The projected iteration converges in
\(\mathcal{O}(\log(1/\epsilon))\) iterations. Given the empirical response
estimates, reconstructing all coefficients requires
\(\widetilde{\mathcal{O}}(M/\epsilon)\) classical operations in the
symbolic-Pauli model. Including the aggregation of measurement outcomes, the
overall classical post-processing complexity is
\(\widetilde{\mathcal{O}}(M/\epsilon^2)\).
\end{theorem}

The guarantees above concern coordinate estimation. If a physical
Lindbladian estimate is required, run the learning procedure with target
coordinate accuracy $\epsilon/2$ and apply the $\ell_\infty$-projection of
Corollary~\ref{cor:optional-physical-projection}. The optional
semidefinite-programming cost is not included in the stated post-processing complexity.

\section{Discussion}
Hamiltonian learning provides a powerful framework for characterizing coherent interactions of closed quantum many-body systems~\citep{gebhartLearningQuantumSystems2023,haahLearningQuantumHamiltonians2024,HTFS2023}. Realistic quantum systems, however, are rarely perfectly isolated: coupling to uncontrolled degrees of freedom gives rise to decoherence, relaxation, and other noise processes, whose time-homogeneous Markovian limit is described by a Lindbladian~\citep{gorini1976completely,lindblad1976generators,breuer2002theory}.  In this setting, conventional Hamiltonian-learning responses no longer identify Hamiltonian coefficients uniquely because the same local response can contain both coherent and dissipative contributions.  Our work resolves this significant challenge at the level of experimentally accessible responses.  By treating the raw Pauli responses as a linear measurement system and
inverting their support-structured mixing, we construct response
coordinates that separately identify all coherent and dissipative
coefficients of a structured Pauli--GKSL generator. The resulting protocols
require only simple state preparation, short-time evolution, and
Pauli measurements, which are experimentally accessible.  For example, stroboscopic dynamics in which short
coherent evolution steps are interleaved with weak memoryless local noise
can approach a Lindblad evolution in a continuous-time limit~\citep{breuer2002theory}; our
framework can then learn the desired interactions and the local noise
rates simultaneously, rather than absorbing the latter into an
effective Hamiltonian. 
More broadly, these results extend scalable generator learning from closed to
Markovian open quantum systems and provide a route toward the calibration of
noisy quantum simulators, the diagnosis of local error mechanisms, and the
verification of engineered dissipative dynamics
~\citep{diehl2008quantumstates,verstraete2009dissipation}.
Extending this response-based approach to time-dependent generators and
non-Markovian environments offers a natural direction for future work
~\citep{FrancaMobusRouzeWerner2025,breuer2016colloquium}.

\emph{Note added.---}
During the completion of this manuscript, we became aware of several
concurrent works on Lindbladian learning and structure learning
~\citep{BirkeEtAl2026,IvashkovEtAl2026,HeightmanEtAl2026,
RomanovEtAl2026,AradChenGuoRebentrostYu2026,
MobusBergamaschiFrancaRouze2026,LewisTangWright2026}.
These works study related learning problems under different assumptions on
prior structural information, locality, and experimental access.
Our setting is complementary: we assume a known nonredundant
Pauli--GKSL candidate dictionary and focus on removing the
coherent--dissipative and common-extension mixing of local Pauli responses
before performing finite-time coefficient reconstruction.
A detailed comparison is provided in
Appendix~\ref{app:comparison}.


\vspace{10pt}
\emph{Acknowledgments.}
X.~Yuan acknowledges support from the National Natural Science Foundation of China Grant No.~12361161602 and NSAF Grant No.~U2330201, the Quantum Science and Technology--National Science and Technology Major Project No.~2023ZD0300200, Beijing Natural Science Foundation Z250004, and Beijing Science and Technology Planning Project Grant No.~Z25110100810000. Y.~Wu acknowledges support from NSFC Grants No.~62501060, No.~62461160263, and No.~62371050.

%

\clearpage
\widetext

\appendix

\section{Comparison with related and concurrent work}
\label{app:comparison}

Recent progress on learning open quantum dynamics has followed several
different directions, depending on whether the goal is reconstructing the
full dynamical map, identifying the interaction structure, or estimating the
parameters of a prescribed generator model.  Our work belongs to the third
category: we assume a known nonredundant local Pauli--GKSL dictionary and
develop a scalable procedure for estimating all coefficients of the generator.
Within this setting, the main difficulty is not identifying which terms are
present, but resolving the intrinsic ambiguity of local responses caused by
the coexistence of Hamiltonian and dissipative contributions.

Several recent works have investigated Lindbladian learning without assuming a
known support structure.  Shadow-based approaches provide efficient
measurement schemes for recovering local properties of unknown dynamics and,
in some cases, allow interaction structure to be inferred from randomized
measurements~\citep{BirkeEtAl2026,AradChenGuoRebentrostYu2026}.  These methods
are designed for settings where the unknown generator itself is sparse or
where the support of the interactions must be discovered.  In contrast, our
framework assumes that the candidate Pauli--GKSL dictionary is supplied as
prior information.  This allows us to focus on a different bottleneck: even
when the possible generator terms are known, a local Pauli response generally
corresponds to a linear combination of several coherent and dissipative
coordinates.  We remove this ambiguity by constructing a response matrix
whose inverse directly produces coefficient-selective normalized responses.

Another line of work studies learning from local dynamical information under
locality assumptions expressed through interaction strengths or quasi-local
structure.  In particular, bounded-strength assumptions have been used to
establish scalable learning guarantees for general local generators
~\citep{IvashkovEtAl2026,MobusBergamaschiFrancaRouze2026,LewisTangWright2026}.
These results emphasize controlling the propagation of information and the
stability of local estimation procedures.  Our locality assumptions play a
related role in controlling finite-time expansions, but the purpose is
different: we use locality not only to bound the dynamics, but also to prove
that the inverse response transformation remains local and well-conditioned.
The resulting normalized responses separate generator-level identifiability
from the subsequent finite-time reconstruction problem.

Several concurrent works consider alternative parameterizations or estimation
strategies.  Romanov \emph{et al.} introduce encoded stabilizer constructions
to suppress the influence of unknown strong interactions
~\citep{RomanovEtAl2026}, while Heightman \emph{et al.} formulate Lindbladian
identification as a data-driven differential-equation fitting problem
~\citep{HeightmanEtAl2026}.  These approaches address different experimental
constraints and do not rely on the explicit Pauli--GKSL response structure
used here.  Our protocol instead requires only product-state preparation,
short-time evolution, and Pauli measurements, and the classical reconstruction
is based on the known generator dictionary rather than a learned surrogate
model.

The closest conceptual comparison is with works that reconstruct generators
through local observables or shadow-based channel information
~\citep{BirkeEtAl2026,MobusBergamaschiFrancaRouze2026,LewisTangWright2026}.
Those approaches typically control the statistical and locality aspects of
the estimation problem.  Here we identify and resolve an additional algebraic
obstruction specific to Lindbladian learning: local responses are not naturally
aligned with individual GKSL coefficients.  By exploiting the common-extension
structure of Pauli responses, we show that the mixing matrix has a
support-ordered triangular structure and construct its inverse with
system-size-independent conditioning under bounded-overlap assumptions.
This yields normalized response coordinates satisfying $\partial_t g_\alpha(0,\theta)=\theta_\alpha$ after which standard finite-time reconstruction techniques can be applied.

Therefore, the present work is complementary to existing approaches rather
than a replacement for them.  Methods based on support discovery, shadow
tomography, or flexible generator fitting are valuable when the microscopic
model is unknown.  Our framework is advantageous when a physically motivated
candidate dictionary is available, as is common in quantum simulation,
device calibration, and engineered dissipative dynamics, where the goal is to
estimate interaction strengths and noise parameters simultaneously.

\section{Local Pauli--GKSL parameterization with a known candidate dictionary}
\label{app:param}

We use a real, nonredundant Pauli--GKSL parameterization on $N$ qubits, with
Hilbert-space dimension $D=2^N$.  For a region $W\subseteq[N]$, $\cP_W$
denotes the phase-free Pauli strings supported inside $W$, tensored with
identities outside $W$.  We write $I$ for the identity Pauli string and $\Id$
for the identity superoperator.

A time-independent Markovian generator \(\cL_\theta\) has the GKSL form
\begin{equation}
  \frac{d\rho}{dt}
  =
  \cL_\theta(\rho)
  :=
  -i[H,\rho]
  +
  \sum_{A,B\in\cB}K_{AB}
  \left(A\rho B-\frac12\{BA,\rho\}\right),
  \label{eq:app-gksl}
\end{equation}
where $\cB$ is a chosen set of traceless local Pauli strings and
$K=K^\dagger\succeq0$ is the Kossakowski matrix.  Excluding the identity from
$\cB$ removes the trivial term $\cD_{I,I}=0$ and the usual gauge
redundancy in which dissipative terms with one identity factor can be absorbed
into the Hamiltonian part.

The Hamiltonian is a real linear combination of the prescribed nonidentity
phase-free Pauli candidates:
\begin{equation}
    H=\sum_C h_C C,
    \qquad h_C\in\R,
    \qquad C\in\cP_{[N]}\setminus\{I\},
    \label{eq:app-H-param}
\end{equation}
where the sum ranges only over the prescribed Hamiltonian candidates.
For the dissipative coordinates, diagonal entries are represented by the real
coefficients $K_{AA}\ge0$.  For an ordered pair $A<B$, we write
\begin{equation}
    k^{\rm R}_{AB}=2\operatorname{Re}K_{AB},
    \qquad
    k^{\rm I}_{AB}=-2\operatorname{Im}K_{AB}.
    \label{eq:app-real-gksl-coords}
\end{equation}
Define
\begin{align}
  \cD_{A,B}^{\rm Re}(\rho)
  &=
  \frac12\left(\cD_{A,B}(\rho)+\cD_{B,A}(\rho)\right),
  \label{eq:app-D-Re}\\
  \cD_{A,B}^{\rm Im}(\rho)
  &=
  \frac1{2i}\left(\cD_{A,B}(\rho)-\cD_{B,A}(\rho)\right),
  \label{eq:app-D-Im}
\end{align}
where
\begin{equation}
    \cD_{A,B}(\rho)=A\rho B-\frac12\{BA,\rho\}.
    \label{eq:app-DAB-sch}
\end{equation}
The relation $K_{BA}=\overline{K_{AB}}$ gives
\begin{equation}
    K_{AB}\cD_{A,B}+K_{BA}\cD_{B,A}
    =
    k^{\rm R}_{AB}\cD_{A,B}^{\rm Re}
    +k^{\rm I}_{AB}\cD_{A,B}^{\rm Im}.
    \label{eq:app-real-imag-combination}
\end{equation}

In the Heisenberg picture, the elementary terms are
\begin{equation}
    \cL_C^{H,\dagger}(O)=i[C,O]=iCO-iOC,
    \label{eq:app-H-adjoint}
\end{equation}
and
\begin{equation}
    \cD_{A,B}^\dagger(O)=BOA-\frac12\{BA,O\}.
    \label{eq:app-D-adjoint}
\end{equation}
Because the Hilbert--Schmidt adjoint is conjugate-linear in scalar
coefficients, the two real-coordinate terms have adjoints
\begin{align}
  \cD_{A,B}^{\rm Re,\dagger}
  &=
  \frac12\left(\cD_{A,B}^\dagger+\cD_{B,A}^\dagger\right),
  \label{eq:app-D-Re-adjoint}\\
  \cD_{A,B}^{\rm Im,\dagger}
  &=
  \frac1{2i}\left(\cD_{B,A}^\dagger-\cD_{A,B}^\dagger\right).
  \label{eq:app-D-Im-adjoint}
\end{align}

Let $a$ index the normalized elementary terms, denoted by
$\mathcal L_a^\dagger$.

\begin{lemma}[Explicit norms of normalized Pauli--GKSL terms]
\label{lem:normalized-term-norms}
Every nonzero elementary term in the fixed Pauli--GKSL normalization obeys
\begin{equation}
    1
    \le
    \|\mathcal L_a^\dagger\|_{\infty\to\infty}
    \le
    2.
    \label{eq:normalized-term-norm-range}
\end{equation}
The Hamiltonian and diagonal dissipative terms have norm exactly $2$.
\end{lemma}

\begin{proof}
The upper bound follows from the triangle estimates for the commutator and
Eq.~\eqref{eq:app-D-adjoint}; taking the real or imaginary half-sum does not
increase the bound $2$.  If a phase-free Pauli $O$ anticommutes with $C$ or
$A$, respectively, then
$\|i[C,O]\|_\infty=2$ and
$\|AOA-O\|_\infty=2$, proving the two exact claims.

For $A\ne B$, write $AB=\varepsilon_{AB}BA$ with
$\varepsilon_{AB}\in\{\pm1\}$.  Nondegeneracy of the Pauli symplectic form
allows a phase-free Pauli $O$ with any prescribed pair of commutation signs
with the distinct strings $A$ and $B$.  If $\varepsilon_{AB}=1$, choose $O$
anticommuting with both strings for the real term and commuting with $A$
but anticommuting with $B$ for the imaginary term; direct substitution
in Eqs.~\eqref{eq:app-D-adjoint}--\eqref{eq:app-D-Im-adjoint} gives output norm
$2$ and $1$, respectively.  If $\varepsilon_{AB}=-1$, interchanging the two
commutation-sign choices gives output norm $1$ and $2$.  Thus every nonzero
off-diagonal real or imaginary term has induced norm at least $1$.
\end{proof}

Let $\cA=\{\cL_\alpha^\dagger\}_{\alpha=1}^M$ be a known linearly independent
real Pauli--GKSL candidate dictionary; candidate coefficients may be zero.  The
unknown adjoint generator is
\begin{equation}
    \cL_\theta^\dagger
    =
    \sum_{\alpha=1}^M\theta_\alpha\cL_\alpha^\dagger,
    \qquad
    \theta\in\R^M.
    \label{eq:app-theta-adjoint}
\end{equation}
More generally, for any real coordinate vector \(x\in\R^M\), write
\begin{equation}
    \cL_x^\dagger
    :=
    \sum_{\alpha=1}^Mx_\alpha\cL_\alpha^\dagger.
    \label{eq:app-generic-coordinate-generator}
\end{equation}
Let $K(\theta)$ be the Hermitian Kossakowski matrix obtained from the
dissipative coordinates through
Eq.~\eqref{eq:app-real-gksl-coords}, with entries absent from the candidate
dictionary fixed to zero.  The physical parameter set is
$\Theta_{\rm GKSL}:=\{\theta\in[-1,1]^M:K(\theta)\succeq0\}$.  The true
parameter belongs to this set, while the algorithms estimate coordinates in
$\R^M$.  Corollary~\ref{cor:optional-physical-projection} gives an optional
projection onto $\Theta_{\rm GKSL}$.

The support of a candidate term is defined as follows:
\begin{align}
    \supp(\cL_C^{H,\dagger})
    &=\supp(C),\\
    \supp(\cD_{A,A}^\dagger)
    &=\supp(A),\\
    \supp(\cD_{A,B}^{\rm Re,\dagger})
    =\supp(\cD_{A,B}^{\rm Im,\dagger})
    &=\supp(A)\cup\supp(B).
\end{align}
Let $\cA_D\subseteq\cA$ be the set of candidate dissipative terms, and define
\begin{equation}
    S_D
    =
    \max_{a\in\cA_D}
    |\supp(\cL_a^\dagger)|,
    \label{eq:dissipative-support-SD}
\end{equation}
with $S_D=0$ when $\cA_D$ is empty.
We denote the candidate support bound by
\begin{equation}
    s=\max_{\alpha\in\cA}|\supp(\cL_\alpha^\dagger)|.
    \label{eq:app-candidate-support-bound}
\end{equation}
For the hybrid family of real raw responses defined in
Sec.~\ref{app:responses}, all prepared and measured Pauli observables have
support size at most $s$.

\section{Raw Pauli responses and the candidate response matrix}
\label{app:responses}

For $R,S\in\cP_W$, define the raw response functional
\begin{equation}
  F^W_{R,S}(\Phi)
  :=
  \frac1{4^{|W|}}\sum_{P\in\cP_W}
  \frac1D
  \Tr\left[(RPS)^\dagger\Phi(P)\right].
  \label{eq:app-raw-response}
\end{equation}
The functional is linear in $\Phi$.  Sec.~\ref{app:measurement} gives an
unbiased estimator of $F^W_{R,S}(e^{t\cL_\theta^\dagger})$.  Because the
response is local, it may also contain coefficients with identical Pauli
extensions outside $W$; Lemma~\ref{lem:app-global-response-formula} makes this
dependence explicit.

\begin{lemma}[Full-system response formula]
\label{lem:app-global-response-formula}
Let a full-system superoperator have the expansion
\begin{equation}
    \Phi(O)=\sum_{\widetilde R,\widetilde S\in\cP_{[N]}}
    a_{\widetilde R,\widetilde S}\,
    \widetilde R O \widetilde S.
    \label{eq:app-global-expansion}
\end{equation}
For a region $W$ and Pauli strings $R,S\in\cP_W$,
\begin{equation}
    F^W_{R,S}(\Phi)
    =
    \sum_{T\in\cP_{W^c}}
    a_{R\otimes T,\,S\otimes T}.
    \label{eq:app-global-response-formula}
\end{equation}
\end{lemma}

\begin{proof}
The average over Pauli strings on $W$ forces the restrictions of
$\widetilde R$ and $\widetilde S$ to $W$ to be $R$ and $S$, respectively.  The
normalized trace over $W^c$ forces the two outside restrictions to be equal.
Summing over this common outside Pauli string gives
Eq.~\eqref{eq:app-global-response-formula}.
\end{proof}

\begin{lemma}[Reality relation for averaged Pauli responses]
\label{lem:app-response-conjugacy}
If $\Phi$ is Hermiticity preserving, then
\begin{equation}
    F^W_{S,R}(\Phi)=\overline{F^W_{R,S}(\Phi)}
    \label{eq:app-response-conjugacy}
\end{equation}
for all $R,S\in\cP_W$.
\end{lemma}

\begin{proof}
For every phase-free Pauli string $P$, the operator $\Phi(P)$ is Hermitian.
Consequently,
\begin{align}
    \overline{\Tr[(RPS)^\dagger\Phi(P)]}
    &=\Tr[\Phi(P)RPS]\\
    &=\Tr[(SPR)^\dagger\Phi(P)],
\end{align}
where the second equality uses cyclicity of the trace.  Averaging over
$P\in\cP_W$ proves Eq.~\eqref{eq:app-response-conjugacy}.
\end{proof}

Let $\operatorname{HP}$ denote the real vector space of
Hermiticity-preserving maps.  All response functionals below are real-linear
on $\operatorname{HP}$.  Lemma~\ref{lem:app-global-response-formula} accounts
for larger common extensions, while
Theorem~\ref{thm:direct-long-support-H-response} treats Hamiltonian terms whose
support exceeds $S_D$.

\paragraph{Direct response for a long-support Hamiltonian candidate.}
Let $C$ be a candidate Hamiltonian Pauli string with
$|\supp(C)|>S_D$.  Choose a site $u\in\supp(C)$ and a one-site
phase-free Pauli string $P_C$ supported on $u$ such that
$\{P_C,C\}=0$.  Define the phase-free Pauli string $Q_C$, the sign
$\chi_C\in\{\pm1\}$, and the direct raw response by
\begin{align}
    iCP_C&=\chi_CQ_C,
    \label{eq:direct-response-QC}\\
    \mathcal R_C(\Phi)
    &:=
    \frac{\chi_C}{2D}\Tr[Q_C\Phi(P_C)].
    \label{eq:direct-response-RC}
\end{align}

\begin{theorem}[Biorthogonality of direct Hamiltonian responses]
\label{thm:direct-long-support-H-response}
For every candidate Hamiltonian Pauli string $C'$ and every candidate
real dissipative term $\cL_\beta^\dagger$, the direct response associated
with $C$ satisfies
\begin{equation}
    \mathcal R_C(\cL_{C'}^{H,\dagger})=\delta_{C,C'},
    \qquad
    \mathcal R_C(\cL_\beta^\dagger)=0.
    \label{eq:direct-response-duality}
\end{equation}
\end{theorem}

\begin{proof}
Since $C$ and $P_C$ anticommute,
\begin{equation}
    \cL_C^{H,\dagger}(P_C)
    =i[C,P_C]
    =2iCP_C
    =2\chi_CQ_C.
\end{equation}
Since $Q_C^2=I$, $\mathcal R_C(\cL_C^{H,\dagger})=1$.  For
$C'\ne C$, either $C'$ commutes with $P_C$, giving zero, or
$\cL_{C'}^{H,\dagger}(P_C)$ is proportional to $C'P_C$.  Its Pauli overlap
with $Q_C\propto CP_C$ can be nonzero only if $C'=C$, which proves the
Hamiltonian identity.

For an elementary dissipative term
\begin{equation}
    \cD_{A,B}^\dagger(O)
    =
    BOA-\frac12\{BA,O\},
\end{equation}
set
\(
V=\supp(A)\cup\supp(B)
\).
If \(u\notin V\), then \(P_C\) commutes with \(A\), \(B\), and \(BA\).
Consequently,
\begin{equation}
    BP_CA=P_CBA,
    \qquad
    \{BA,P_C\}=2P_CBA,
\end{equation}
and hence
\begin{equation}
    \cD_{A,B}^\dagger(P_C)=0.
\end{equation}

Now suppose that \(u\in V\).  Since \(P_C\) is supported only on \(u\),
we have \(\supp(P_C)\subseteq V\).  Every term in
\(\cD_{A,B}^\dagger(P_C)\) is therefore supported inside \(V\), and hence
\begin{equation}
    \supp\!\left(\cD_{A,B}^\dagger(P_C)\right)
    \subseteq V,
    \qquad
    |V|\le S_D.
    \label{eq:direct-response-dissipative-support}
\end{equation}
On the other hand, multiplication by the one-site Pauli \(P_C\) does not
remove any site from the support of \(C\), because \(P_C\) anticommutes
with \(C\) at \(u\).  Thus
\begin{equation}
    \supp(Q_C)=\supp(C),
    \qquad
    |\supp(Q_C)|>S_D.
\end{equation}
It follows that every Pauli string appearing in
\(\cD_{A,B}^\dagger(P_C)\) is distinct from \(Q_C\).  Pauli
orthogonality therefore gives
\begin{equation}
    \Tr\!\left[
        Q_C\cD_{A,B}^\dagger(P_C)
    \right]
    =0.
\end{equation}
The same conclusion holds for
\(\cD_{A,B}^{\rm Re,\dagger}\) and
\(\cD_{A,B}^{\rm Im,\dagger}\) by linearity.
\end{proof}

Let $\cJ_\cA:=\{1,\ldots,M\}$ be the index set for a square real raw-response
family \(\{\cF_j:j\in\cJ_\cA\}\), containing one bounded real response for
each candidate coordinate.  Its four response types are as follows.
\begin{enumerate}
\item For a Hamiltonian coefficient $h_C$ with $|\supp(C)|\le S_D$, use
\begin{equation}
    \cF_C^H(\Phi)
    :=
    \frac1{2i}\left(
        F^{\supp(C)}_{C,I}(\Phi)
        -F^{\supp(C)}_{I,C}(\Phi)
    \right).
    \label{eq:app-H-raw-response-row}
\end{equation}

\item For $|\supp(C)|>S_D$, use the direct response
\begin{equation}
    \cF_C^H:=\mathcal R_C
    \label{eq:app-H-direct-real-row}
\end{equation}
from Eq.~\eqref{eq:direct-response-RC}.

\item For a diagonal dissipative coefficient $K_{AA}$, use
\begin{equation}
    \cF_A^{\rm diag}:=F^{\supp(A)}_{A,A}.
    \label{eq:app-diag-raw-response-row}
\end{equation}

\item For an ordered off-diagonal pair $A<B$, set
\begin{equation}
    W_{AB}=\supp(A)\cup\supp(B),
    \qquad
    F_{BA}=F^{W_{AB}}_{B,A},
    \qquad
    F_{AB}=F^{W_{AB}}_{A,B}.
    \label{eq:app-offdiag-raw-response-rows}
\end{equation}
Define the real off-diagonal raw responses
\begin{align}
    \cF_{AB}^{\rm Re}(\Phi)
    &:=
    \frac12\left(F_{BA}(\Phi)+F_{AB}(\Phi)\right),
    \label{eq:app-real-offdiag-raw-Re}\\
    \cF_{AB}^{\rm Im}(\Phi)
    &:=
    -\frac{i}{2}\left(F_{BA}(\Phi)-F_{AB}(\Phi)\right).
    \label{eq:app-real-offdiag-raw-Im}
\end{align}
If both $k^{\rm R}_{AB}$ and $k^{\rm I}_{AB}$ are candidate, include both real rows.  If only one
coordinate is candidate, include only its corresponding real row.  By
Lemma~\ref{lem:app-response-conjugacy}, all of
$\cF_C^H$, $\cF_A^{\rm diag}$, $\cF_{AB}^{\rm Re}$,
$\cF_{AB}^{\rm Im}$, and $\mathcal R_C$ are real on every
Hermiticity-preserving map used in the protocol.
\end{enumerate}

The local row normalizations remain valid in the presence of larger-support
contamination.  For a Hamiltonian term,
\begin{equation}
    \cF_C^H(\cL_C^{H,\dagger})=1.
    \label{eq:app-H-local-evals}
\end{equation}
For an off-diagonal pair with both real coordinates candidate, the local real
block is
\begin{equation}
    \begin{pmatrix}
    \cF_{AB}^{\rm Re}(\cD_{A,B}^{\rm Re,\dagger}) &
    \cF_{AB}^{\rm Re}(\cD_{A,B}^{\rm Im,\dagger}) \\
    \cF_{AB}^{\rm Im}(\cD_{A,B}^{\rm Re,\dagger}) &
    \cF_{AB}^{\rm Im}(\cD_{A,B}^{\rm Im,\dagger})
    \end{pmatrix}
    =
    \begin{pmatrix}
    \frac12 & 0 \\
    0 & \frac12
    \end{pmatrix},
    \label{eq:app-offdiag-local-block}
\end{equation}
which is invertible.  The row combinations
$2\cF_{AB}^{\rm Re}$ and $2\cF_{AB}^{\rm Im}$
are unit rows for $\cD_{A,B}^{\rm Re,\dagger}$ and
$\cD_{A,B}^{\rm Im,\dagger}$, respectively.  These local normalizations remain
valid before the larger common-extension terms are removed by the global
inverse below.

In the remainder of the Supplementary Information, $\cF_j$ for
$j\in\cJ_\cA$ denotes a generic element of this hybrid family of real raw
responses.  Let $\mathbf T$ be the fixed diagonal row-normalization matrix
that multiplies the real-part and imaginary-part off-diagonal response rows by
two and leaves all other rows unchanged.  Thus
$\|\mathbf T\|_{\infty\to\infty}\le2$.

Define the candidate response matrix by
\begin{equation}
    \mathbf G_{j\beta}=\cF_j(\cL_\beta^\dagger)\in\R,
    \qquad
    j\in\cJ_\cA,
    \quad
    \beta\in\cA.
    \label{eq:app-response-G}
\end{equation}
All entries of $\mathbf G$ are known from the candidate dictionary.

\begin{proposition}[Invertibility of the candidate response matrix]
\label{prop:app-normalized-coordinates}
The square response matrix $\mathbf G\in\R^{M\times M}$ is invertible.
\end{proposition}

\begin{proof}
Set $\widetilde{\mathbf G}=\mathbf T\mathbf G$.  Order dissipative columns
before Hamiltonian columns and order dissipative terms by decreasing
support size.  Dissipative response rows vanish on Hamiltonian columns:
Hamiltonian terms contain only the one-sided left--right terms $(C,I)$
and $(I,C)$, whereas the primary dissipative rows have nonidentity Pauli
strings on both sides.  By Lemma~\ref{lem:app-global-response-formula}, a
dissipative row can therefore mix only with dissipative columns related by a
common outside extension, such as $A\mapsto AT$ or
$(A,B)\mapsto(AT,BT)$.  The chosen ordering makes the dissipative block
triangular with invertible unit diagonal blocks.

For the Hamiltonian block, $\cF_C^H$ evaluates to one on
$\cL_C^{H,\dagger}$ and to zero on every different Hamiltonian term.
The identity side excludes larger-support Hamiltonian common extensions.  A
short-support Hamiltonian row may receive known dissipative contamination,
whereas Theorem~\ref{thm:direct-long-support-H-response} shows that every
long-support direct row has none.  Consequently,
\begin{equation}
    \widetilde{\mathbf G}
    =
    \begin{pmatrix}
        \mathbf G_D & 0 \\
        * & I_H
    \end{pmatrix},
    \label{eq:app-candidate-block-triangular}
\end{equation}
where $\mathbf G_D$ is triangular with invertible diagonal blocks.  Hence
$\widetilde{\mathbf G}$, and therefore $\mathbf G$, is invertible.
\end{proof}

Define the response inverse and the normalized response functionals by
\begin{equation}
    \mathbf H:=\mathbf G^{-1}\in\R^{M\times M},
    \qquad
    \mathbf H\mathbf G=I_M,
\end{equation}
\begin{equation}
    \mathsf C_\alpha(\Phi)
    :=\sum_{j\in\cJ_\cA}(\mathbf H)_{\alpha j}\cF_j(\Phi),
    \qquad
    \mathsf C_\alpha:\operatorname{HP}\to\R.
    \label{eq:app-C-alpha-left-inverse}
\end{equation}
Write
\begin{equation}
    L_C
    :=
    \max_\alpha\|(\mathbf H)_{\alpha\cdot}\|_1
    =\|\mathbf H\|_{\infty\to\infty},
\end{equation}
and define the maximum row sparsity by
\begin{equation}
    r_0
    :=
    \max_\alpha
    \left|
    \left\{j:(\mathbf H)_{\alpha j}\ne0\right\}
    \right|.
    \label{eq:response-inverse-row-sparsity}
\end{equation}
The total number of nonzero entries is
\begin{equation}
    \operatorname{nnz}(\mathbf H)
    :=
    \left|
    \left\{(\alpha,j):(\mathbf H)_{\alpha j}\ne0\right\}
    \right|.
    \label{eq:app-response-inverse-nnz}
\end{equation}
For a long-support Hamiltonian coordinate $\alpha$ associated with the Pauli
string $C$, the direct response row is an isolated unit coordinate in the
block form Eq.~\eqref{eq:app-candidate-block-triangular}.  Consequently,
\begin{equation}
    \mathsf C_\alpha=\mathcal R_C,
    \qquad
    |\supp(C)|>S_D.
    \label{eq:app-long-support-normalized-functional}
\end{equation}

We can therefore define the normalized response coordinates uniformly over
the full candidate dictionary by
\begin{equation}
    g_\alpha(t,\theta)
    :=\mathsf C_\alpha\!\left(e^{t\cL_\theta^\dagger}-\Id\right),
    \qquad
    g(t,\theta):=(g_\alpha(t,\theta))_{\alpha\in\cA}.
    \label{eq:app-normalized-response-function}
\end{equation}
For a long-support Hamiltonian coordinate $\alpha$ associated with $C$, this
definition explicitly specializes to
\begin{equation}
    g_\alpha(t,\theta)
    =
    \mathcal R_C\!\left(e^{t\cL_\theta^\dagger}-\Id\right),
    \qquad
    |\supp(C)|>S_D.
    \label{eq:app-long-support-normalized-response}
\end{equation}

\begin{corollary}[Unified first-order expansion of the normalized response map]
\label{cor:app-normalized-response-first-order}
For every coefficient vector $\theta\in\R^M$,
\begin{equation}
    g(t,\theta)
    =
    t\theta+\mathcal{O}(t^2)
    \qquad
    (t\to0),
    \label{eq:app-normalized-response-first-order}
\end{equation}
where the remainder is understood in the vector $\ell_\infty$ norm.
\end{corollary}

\begin{proof}
For each $\alpha\in\cA$, expanding
Eq.~\eqref{eq:app-normalized-response-function} at $t=0$ and using
$\cL_\theta^\dagger=\sum_{\beta\in\cA}
\theta_\beta\cL_\beta^\dagger$ gives
\begin{align}
    g_\alpha(t,\theta)
    &=
    t\,\mathsf C_\alpha(\cL_\theta^\dagger)
    +\mathcal{O}(t^2) \notag\\
    &=
    t\sum_{\beta\in\cA}\theta_\beta
    \sum_{j\in\cJ_\cA}
    (\mathbf H)_{\alpha j}
    \cF_j(\cL_\beta^\dagger)
    +\mathcal{O}(t^2) \notag\\
    &=
    t\sum_{\beta\in\cA}\theta_\beta
    (\mathbf H\mathbf G)_{\alpha\beta}
    +\mathcal{O}(t^2)
    =
    t\theta_\alpha+\mathcal{O}(t^2).
\end{align}
The long-support Hamiltonian case is included in the same calculation through
Eqs.~\eqref{eq:app-long-support-normalized-functional}
and~\eqref{eq:app-long-support-normalized-response}.  Since the dictionary is
finite, collecting the componentwise remainders in the $\ell_\infty$ norm
proves Eq.~\eqref{eq:app-normalized-response-first-order}.
\end{proof}

\begin{remark}[Role of the square construction]
\label{rem:app-no-r0}
The square response family contains exactly $M$ raw responses, so the union
bound in Sec.~\ref{app:measurement} involves $Mn_t$ response--time pairs.
This counting fact does not imply that the inverse $\mathbf{H}$ is sparse.  If
$\mathbf{H}$ is dense, the classical multiplication by $\mathbf{H}$ may cost $\mathcal{O}(M^2)$ per time
point.  Theorem~\ref{thm:uniform-response-inverse} gives sufficient
bounded-overlap conditions under which $r_0=\mathcal{O}(1)$ and hence
$\operatorname{nnz}(\mathbf{H})=\mathcal{O}(M)$, so the post-processing remains linear up
to local constants.
\end{remark}

\section{Measurement of raw and normalized responses}
\label{app:measurement}

To estimate an averaged raw response
$F^W_{R,S}(e^{t\cL_\theta^\dagger})$, use the following protocol.
\begin{enumerate}
\item Sample $P\in\cP_W$ and write
\begin{equation}
    RPS=\omega(P)Q(P),
    \qquad
    \omega(P)\in\{\pm1,\pm i\},
    \label{eq:app-RPS-decomp}
\end{equation}
where $Q(P)$ is phase-free and Hermitian.

\item If $Q(P)\ne I$, sample $\sigma\in\{\pm1\}$ uniformly and prepare an
ensemble with mean state
\begin{equation}
    \rho_{\sigma,Q(P)}=\frac{I+\sigma Q(P)}{D}.
    \label{eq:app-pauli-mixture}
\end{equation}
Evolve for time $t$, measure $P$ with outcome $\mu\in\{\pm1\}$, and record
\begin{equation}
    Z=\omega(P)^*\sigma\mu.
    \label{eq:app-Z}
\end{equation}
If $Q(P)=I$, omit the random sign, use a maximally mixed input, and record
$Z=\omega(P)^*\mu$.

\item For a direct response $\mathcal R_C$, replace $Q(P)$ by $Q_C$, sample
$\sigma\in\{\pm1\}$, and prepare
\begin{equation}
    \rho_{\sigma,Q_C}=\frac{I+\sigma Q_C}{D},
\end{equation}
Evolve for time $t$, measure $P_C$ with outcome $\mu_C\in\{\pm1\}$, and record
\begin{equation}
    Z_C=\frac{\chi_C}{2}\sigma\mu_C.
    \label{eq:app-ZC}
\end{equation}
\end{enumerate}

When $Q(P)$ has multiqubit support, realize its mixed state by sampling product
eigenstates of the nonidentity single-qubit Pauli factors, with local
eigenvalues uniform subject to product $\sigma$; sample all remaining qubits
from any product ensemble with one-qubit mean $I/2$.  Use the same realization
with $Q_C$ for a direct response.  The proof below depends only on the
ensemble-average input state.

\begin{theorem}[Unbiased raw response estimation]
\label{thm:app-unbiased}
Let $\cE_t=e^{t\cL_\theta}$ and
$\cE_t^\dagger=e^{t\cL_\theta^\dagger}$.  The random variable in
Eq.~\eqref{eq:app-Z} satisfies
\begin{equation}
    \E[Z]=F^W_{R,S}(\cE_t^\dagger).
\end{equation}
For every direct raw response in the hybrid set, the variable in
Eq.~\eqref{eq:app-ZC} satisfies
\begin{equation}
    \E[Z_C]=\mathcal R_C(\cE_t^\dagger).
    \label{eq:app-direct-unbiased}
\end{equation}
\end{theorem}

\begin{proof}
By linearity of the channel and the Born rule, averaging the product-state
realizations is equivalent to using their average density matrices in
Eq.~\eqref{eq:app-pauli-mixture}.  Condition on $P$.  Since
$RPS=\omega(P)Q(P)$,
$(RPS)^\dagger=\omega(P)^*Q(P)$.  For $Q(P)\ne I$, conditioning on
$\sigma$ gives
\begin{equation}
 \E[\mu\mid\sigma,P]
 =D^{-1}\Tr[P\cE_t(I)]
 +\sigma D^{-1}\Tr[P\cE_t(Q(P))].
\end{equation}
Multiplication by $\omega(P)^*\sigma$ and averaging over the sign removes the
first term.  The adjoint relation then gives
\begin{equation}
    \E[Z\mid P]
    =D^{-1}\Tr[(RPS)^\dagger\cE_t^\dagger(P)].
\end{equation}
The same expression holds when $Q(P)=I$.  Averaging over the uniform
$P\in\cP_W$ proves the first claim.  For the direct response, averaging over
$\sigma$ gives
\(
 \E[Z_C]=(\chi_C/2D)\Tr[Q_C\cE_t^\dagger(P_C)]
 =\mathcal R_C(\cE_t^\dagger)
\),
which proves Eq.~\eqref{eq:app-direct-unbiased}.
\end{proof}

The variable $Z$ may be complex when $\omega(P)=\pm i$.  The real raw
coordinates in Sec.~\ref{app:responses} use the following single-shot
estimators:
\begin{align}
    Z_C^H
    &:=
    \operatorname{Im}Z_{C,I},
    &
    \E[Z_C^H]
    &=
    \cF_C^H(\cE_t^\dagger),
    \label{eq:app-real-estimator-H}\\
    Z_A^{\rm diag}
    &:=
    \operatorname{Re}Z_{A,A},
    &
    \E[Z_A^{\rm diag}]
    &=
    \cF_A^{\rm diag}(\cE_t^\dagger),
    \label{eq:app-real-estimator-diag}\\
    Z_{AB}^{\rm Re}
    &:=
    \operatorname{Re}Z_{B,A},
    &
    \E[Z_{AB}^{\rm Re}]
    &=
    \cF_{AB}^{\rm Re}(\cE_t^\dagger),
    \label{eq:app-real-estimator-offdiag-Re}\\
    Z_{AB}^{\rm Im}
    &:=
    \operatorname{Im}Z_{B,A},
    &
    \E[Z_{AB}^{\rm Im}]
    &=
    \cF_{AB}^{\rm Im}(\cE_t^\dagger).
    \label{eq:app-real-estimator-offdiag-Im}
\end{align}
Here $Z_{R,S}$ denotes the variable in Eq.~\eqref{eq:app-Z} for the averaged
response $F^W_{R,S}$.  The equalities follow from
Theorem~\ref{thm:app-unbiased} and
Lemma~\ref{lem:app-response-conjugacy}.  Each estimator is real, unbiased, and
bounded by one in magnitude; the direct-response estimator is bounded by
$1/2$.

\begin{proposition}[Simultaneous response-estimation complexity]
\label{prop:app-response-samples}
Assume the square response family indexed by $\cJ_\cA$ and
$\max_\alpha\norm{(\mathbf{H})_{\alpha\cdot}}_1\le L_C$.
For \(0<\delta<1\), if all $M$ normalized responses are needed at $n_t$
evolution times with coordinatewise error at most $\varepsilon_g$, then, with
probability at least \(1-\delta\), it suffices to use
\begin{equation}
    N_{\rm resp}
    =
    \mathcal{O}\!\left(
        M n_t\,
        \frac{L_C^2}{\varepsilon_g^2}
        \log\frac{4M n_t}{\delta}
    \right)
    \label{eq:app-response-sample-complexity}
\end{equation}
state preparations and measurements.
\end{proposition}

\begin{proof}
For every \(j\in\cJ_\cA\), estimate the real raw response \(\cF_j\) at every
required time to accuracy
$\varepsilon_g/L_C$.  Then, for every $\alpha$,
\begin{equation}
    \abs{\widehat{\mathsf C}_\alpha-\mathsf C_\alpha}
    \le
    \sum_{j\in\cJ_\cA}\abs{(\mathbf{H})_{\alpha j}}
    \abs{\widehat{\cF}_j-\cF_j}
    \le
    L_C\frac{\varepsilon_g}{L_C}
    =
    \varepsilon_g.
\end{equation}
Hoeffding's inequality applied directly to the bounded real estimators above,
followed by a union bound over $M n_t$ raw-response/time pairs, gives
Eq.~\eqref{eq:app-response-sample-complexity}.
\end{proof}

\section{Local Taylor bounds}
\label{app:taylor}

We use local support growth to control finite-time response errors.  Every
Pauli observable in the hybrid response family has support size at most $s$,
so Theorem~\ref{thm:app-local-strength} applies with $s_0\le s$.

For a state $\rho$ and observable $O$, write
$R_{\rho,O}(t,\theta)=\Tr[\rho e^{t\cL_\theta^\dagger}(O)]$.  Its Taylor
expansion is
\begin{equation}
    R_{\rho,O}(t,\theta)
    =
    \sum_{k=0}^\infty
    \frac{t^k}{k!}
    \Tr\left[\rho(\cL_\theta^\dagger)^k(O)\right].
    \label{eq:app-scalar-taylor}
\end{equation}
Trace-norm/operator-norm duality bounds its $k$th coefficient by
$\norm{\rho}_1\norm{(\cL_\theta^\dagger)^k(O)}_\infty$; for density matrices,
$\norm{\rho}_1=1$.

Every elementary candidate term has bounded $\infty\to\infty$ norm:
$\norm{i[C,O]}_\infty\le2\norm{O}_\infty$ for a Hamiltonian Pauli term,
whereas for a dissipative term
\begin{equation}
    \norm{BOA-\frac12\{BA,O\}}_\infty
    \le
    \norm{BOA}_\infty+\frac12\norm{BAO}_\infty+
    \frac12\norm{OBA}_\infty
    \le2\norm{O}_\infty.
\end{equation}
The same norm bound holds for the real and imaginary off-diagonal terms.

Every elementary Pauli--GKSL term used below has the following two support
properties.  If $X_a:=\supp(\cL_a^\dagger)$ and $Z$ is local, then
\begin{equation}
    X_a\cap\supp(Z)=\varnothing
    \quad\Longrightarrow\quad
    \cL_a^\dagger(Z)=0,
    \qquad
    \supp(\cL_a^\dagger(Z))
    \subseteq \supp(Z)\cup X_a.
    \label{eq:local-term-support-properties}
\end{equation}

\subsection{Taylor bounds under bounded dual-graph degree}
\label{app:finite-degree-bound}

Let the dual interaction graph have one vertex for each candidate dictionary
term.  Two vertices are adjacent when the corresponding supports
intersect.  Let $d$ be the maximum degree and write
$N(O)=|\{a:\supp(O)\cap\supp(\cL_a^\dagger)\ne\emptyset\}|$.

\begin{theorem}[Mixed Taylor bound from bounded dual-interaction-graph degree]
\label{thm:app-finite-degree}
Assume the support properties in
Eq.~\eqref{eq:local-term-support-properties} and
$\|\cL_a^\dagger\|_{\infty\to\infty}\le2$ for every candidate term.
For $k\ge1$, let $z^{(1)},\ldots,z^{(k)}$ satisfy
$\|z^{(\ell)}\|_\infty\le1$.  If $\|O\|_\infty\le1$ and
$N(O)\le d+1$, then
\begin{equation}
    \left\|
        \cL_{z^{(k)}}^\dagger\cdots
        \cL_{z^{(1)}}^\dagger(O)
    \right\|_\infty
    \le
    [2(d+1)]^k k!.
    \label{eq:app-correct-finite-degree-bound}
\end{equation}
\end{theorem}

\begin{proof}
Expand the mixed product into ordered words.  A nonzero word must grow through
the dual interaction graph.  Its first term has at most $N(O)$ choices,
and after $\ell$ terms have been chosen the next one has at most
$N(O)+\ell d$ choices.  The coefficient bounds and the elementary norm bound
therefore give
\begin{align}
    \left\|
        \cL_{z^{(k)}}^\dagger\cdots
        \cL_{z^{(1)}}^\dagger(O)
    \right\|_\infty
    &\le
    2^kN(O)(N(O)+d)\cdots(N(O)+(k-1)d)\\
    &\le
    [2(d+1)]^k k!.
\end{align}
\end{proof}

\subsection{Taylor bounds under bounded local strength}
\label{app:local-strength-bound}

Define the unweighted local dictionary strength and the maximum body size by
\begin{equation}
    \kappa_0
    :=
    \max_{y\in[N]}
    \sum_{a:\,y\in\supp(\cL_a^\dagger)}
    \norm{\cL_a^\dagger}_{\infty\to\infty},
    \qquad
    s
    :=
    \max_a\abs{\supp(\cL_a^\dagger)}.
    \label{eq:app-KS}
\end{equation}
The quantity $\kappa_0$ depends only on the prescribed dictionary and not on a
coefficient vector.

\begin{theorem}[Repeated-generator bounds under local strength]
\label{thm:app-local-strength}
Assume the support properties in
Eq.~\eqref{eq:local-term-support-properties}.  Let
$s_0:=|\supp(O)|$, assume $\|O\|_\infty\le1$, and let
$\|z\|_\infty\le1$.  Then, for every $k\ge1$,
\begin{equation}
    \|(\cL_z^\dagger)^k(O)\|_\infty
    \le
    \kappa_0^k
    \prod_{\ell=0}^{k-1}(s_0+\ell s).
    \label{eq:app-local-strength-taylor-general}
\end{equation}
If $s_0\le s$, this simplifies to
\begin{equation}
    \|(\cL_z^\dagger)^k(O)\|_\infty
    \le
    k!\,[s\kappa_0]^k.
    \label{eq:app-local-strength-taylor}
\end{equation}
\end{theorem}

\begin{proof}
Expand $(\cL_z^\dagger)^k(O)$ into ordered words.  After a fixed prefix of
$\ell-1$ terms has been chosen, let $V_{\ell-1}$ be the union of
$\supp(O)$ and the supports of those terms.  By
Eq.~\eqref{eq:local-term-support-properties}, a term that produces
a nonzero next step must intersect $V_{\ell-1}$, while
$|V_{\ell-1}|\le s_0+(\ell-1)s$.  Since $\|z\|_\infty\le1$, the total norm
contribution of all possible terms at this step is bounded by
\begin{align}
    \sum_{a_\ell:\,X_{a_\ell}\cap V_{\ell-1}\ne\emptyset}
    |z_{a_\ell}|
    \|\cL_{a_\ell}^\dagger\|_{\infty\to\infty}
    &\le
    \sum_{y\in V_{\ell-1}}
    \sum_{a_\ell:\,y\in X_{a_\ell}}
    \|\cL_{a_\ell}^\dagger\|_{\infty\to\infty}
    \notag\\
    &\le
    (s_0+(\ell-1)s)\kappa_0.
    \label{eq:local-strength-power-one-step}
\end{align}
Multiplying these bounds for $\ell=1,\ldots,k$ proves
Eq.~\eqref{eq:app-local-strength-taylor-general}.  If $s_0\le s$, then
$s_0+\ell s\le(\ell+1)s$, which gives
Eq.~\eqref{eq:app-local-strength-taylor}.
\end{proof}

For results that hold under either locality assumption, define the common
growth scale
\begin{equation}
    \Lambda
    :=
    \begin{cases}
        2(d+1), & \text{bounded dual-interaction-graph degree},\\
        s\kappa_0, & \text{unweighted local strength}.
    \end{cases}
    \label{eq:unified-locality-scale}
\end{equation}
Specializing Theorem~\ref{thm:app-finite-degree} to a repeated coefficient
vector and applying Theorem~\ref{thm:app-local-strength} give
$\|(\cL_z^\dagger)^k(O)\|_\infty\le k!\Lambda^k$ for the raw-response inputs
used below.

\subsection{Uniform bounds on the response inverse}
\label{subsec:response-inverse-bounds}

The constant $L_C$ is determined by the inverse of the candidate response
matrix.  The following regime makes this inverse uniformly local.

We call it the uniform bounded-overlap regime when either locality hypothesis
in Eq.~\eqref{eq:unified-locality-scale} holds and
\begin{equation}
    S_D=\mathcal{O}(1),
    \qquad
    \Lambda=\mathcal{O}(1),
    \label{eq:uniform-response-overlap-assumption}
\end{equation}
with constants independent of the system size $N$ and dictionary size $M$.

\begin{theorem}[Uniform conditioning and sparsity of the response inverse]
\label{thm:uniform-response-inverse}
In the uniform bounded-overlap regime, the inverse
$\mathbf H=\mathbf G^{-1}$ satisfies
\begin{equation}
    L_C
    =
    \max_\alpha\|(\mathbf H)_{\alpha\cdot}\|_1
    =\mathcal{O}(1),
    \label{eq:LC-constant-conclusion}
\end{equation}
and
\begin{equation}
    r_0=\mathcal{O}(1),
    \qquad
    \operatorname{nnz}(\mathbf H)
    \le Mr_0
    =\mathcal{O}(M).
    \label{eq:response-inverse-sparsity-bound}
\end{equation}
All bounds are uniform in $N$ and $M$.
\end{theorem}

\begin{proof}
We first identify the relevant block structure.  Each raw-response functional
is a linear combination of scalar Pauli matrix
elements with total absolute coefficient weight at most one.
Lemma~\ref{lem:normalized-term-norms} therefore gives
\begin{equation}
    \max_{j,\beta}
    \left|(\mathbf T\mathbf G)_{j\beta}\right|
    \le4.
    \label{eq:explicit-response-entry-bound}
\end{equation}
By the square response design, after applying the fixed local row
normalization \(\mathbf T\), the response matrix has the block triangular form
\begin{equation}
    \widetilde{\mathbf G}
    =
    \begin{pmatrix}
        \mathbf{G}_D & 0 \\
        \mathbf{G}_{HD} & I_H
    \end{pmatrix}.
    \label{eq:block-triangular-response-proof}
\end{equation}
Here \(\mathbf{G}_D\) is the dissipative common-extension block, \(I_H\) is the
Hamiltonian block, and \(\mathbf{G}_{HD}\) contains the possible dissipative
contamination of short-support Hamiltonian response rows.  Direct-response
rows have no nonzero entries in $\mathbf{G}_{HD}$, so this block does not
affect invertibility.

Write $\mathbf{G}_D=I+U$, where $U$ is strictly lower triangular after ordering dissipative terms
by decreasing support size.  Each nontrivial common extension strictly increases
the dissipative support.  Since every candidate dissipative term has support
size at most $S_D$, we have $U^{S_D+1}=0$ and therefore
\begin{equation}
    \mathbf{G}_D^{-1}
    =
    \sum_{\ell=0}^{S_D}(-U)^\ell .
    \label{eq:GD-inverse-finite-series-proof}
\end{equation}

This block form immediately yields an inverse-norm bound.  Put
$\Gamma_D=\lVert U\rVert_{\infty\to\infty}$ and
$\Gamma_{HD}=\lVert \mathbf{G}_{HD}\rVert_{\infty\to\infty}$.  Then
\begin{align}
    \lVert \mathbf{G}_D^{-1}\rVert_{\infty\to\infty}
    &\le
    \sum_{\ell=0}^{S_D}\Gamma_D^\ell,
    \label{eq:GD-inverse-row-bound-proof} \\
    \widetilde{\mathbf G}^{-1}
    &=
    \begin{pmatrix}
        \mathbf{G}_D^{-1} & 0 \\
        -\mathbf{G}_{HD}\mathbf{G}_D^{-1} & I_H
    \end{pmatrix},
    \label{eq:block-inverse-proof} \\
    \lVert \widetilde{\mathbf G}^{-1}\rVert_{\infty\to\infty}
    &\le
    \max\left\{
        \sum_{\ell=0}^{S_D}\Gamma_D^\ell,
        1+\Gamma_{HD}\sum_{\ell=0}^{S_D}\Gamma_D^\ell
    \right\}.
    \label{eq:Gtilde-inverse-row-bound-proof}
\end{align}
Since $\mathbf H=\widetilde{\mathbf G}^{-1}\mathbf T=\mathbf G^{-1}$,
$L_C\le\|\widetilde{\mathbf G}^{-1}\|_{\infty\to\infty}
\|\mathbf T\|_{\infty\to\infty}\le
2\|\widetilde{\mathbf G}^{-1}\|_{\infty\to\infty}$.

The same block argument also controls the row supports.  Define
\begin{equation}
    s_{\rm row}(A)
    :=
    \max_i\left|\{j:A_{ij}\ne0\}\right|.
    \label{eq:matrix-row-sparsity-definition}
\end{equation}
The elementary support-counting inequalities
\begin{equation}
    s_{\rm row}(A+B)
    \le s_{\rm row}(A)+s_{\rm row}(B),
    \qquad
    s_{\rm row}(AB)
    \le s_{\rm row}(A)s_{\rm row}(B)
    \label{eq:matrix-row-sparsity-calculus}
\end{equation}
imply, with $b_D:=s_{\rm row}(U)$, that
\begin{equation}
    s_{\rm row}(\mathbf{G}_D^{-1})
    \le
    \sum_{\ell=0}^{S_D}b_D^\ell .
    \label{eq:GD-inverse-row-sparsity-proof}
\end{equation}
Put $b_{HD}:=s_{\rm row}(\mathbf{G}_{HD})$.  The block inverse in
Eq.~\eqref{eq:block-inverse-proof} then gives
\begin{equation}
    s_{\rm row}(\widetilde{\mathbf G}^{-1})
    \le
    \max\left\{
        \sum_{\ell=0}^{S_D}b_D^\ell,
        1+b_{HD}\sum_{\ell=0}^{S_D}b_D^\ell
    \right\}.
    \label{eq:Gtilde-inverse-row-sparsity-proof}
\end{equation}
The row-normalization matrix $\mathbf T$ is diagonal with nonzero diagonal entries, so
right multiplication by $\mathbf T$ does not change a row support.  Therefore the same
bound holds for $\mathbf H=\widetilde{\mathbf G}^{-1}\mathbf T$.

It remains to show that either locality assumption makes the preceding bounds
uniform.  Under bounded dual-interaction-graph degree, every dissipative common extension
contributing to a fixed
dissipative response row must overlap the support of that row.  Likewise, a
dissipative term can contaminate a Hamiltonian response row only when its
support overlaps the Hamiltonian term; if the Hamiltonian support
is larger than \(S_D\), such a dissipative contamination is impossible.  Thus
each relevant row contains at most \(d+1\) candidate dissipative columns.  Hence
\begin{equation}
    b_D,b_{HD}\le d+1,
    \qquad
    \Gamma_D,\Gamma_{HD}\le4(d+1).
    \label{eq:finite-degree-response-row-bounds}
\end{equation}
Since $d+1=\Lambda/2$, these bounds are at most
$\lceil\Lambda\rceil$ and $4\Lambda$, respectively.

In the local-strength setting, consider any set \(V\subseteq[N]\) with
\(|V|\le S_D\).  Lemma~\ref{lem:normalized-term-norms} bounds the number of
candidate terms intersecting \(V\) directly by the local strength:
\begin{align}
    \left|
    \left\{
        a:
        \operatorname{supp}(\mathcal L_a^\dagger)\cap V\ne\emptyset
    \right\}
    \right|
    &\le
    \sum_{\substack{
        a:\operatorname{supp}(\mathcal L_a^\dagger)\cap V\ne\emptyset
    }}
    \lVert \mathcal L_a^\dagger\rVert_{\infty\to\infty}
    \label{eq:local-strength-count-step-one} \\
    &\le
    \sum_{y\in V}
    \sum_{a:y\in\operatorname{supp}(\mathcal L_a^\dagger)}
    \lVert \mathcal L_a^\dagger\rVert_{\infty\to\infty}
    \label{eq:local-strength-count-step-two} \\
    &\le
    \kappa_0|V|
    \le
    \kappa_0S_D.
    \label{eq:local-strength-count-step-three}
\end{align}
Hence every row contributing to \(U\) or \(\mathbf{G}_{HD}\) contains at most
\(\kappa_0S_D\) relevant candidate dissipative columns.  Therefore
\begin{equation}
    b_D,b_{HD}
    \le
    \lceil \kappa_0S_D\rceil,
    \qquad
    \Gamma_D,\Gamma_{HD}
    \le
    4\kappa_0S_D.
    \label{eq:local-strength-response-row-bounds}
\end{equation}
Here $S_D\le s$, so these bounds are also at most
$\lceil\Lambda\rceil$ and $4\Lambda$, respectively.  Substituting these common
bounds into Eq.~\eqref{eq:Gtilde-inverse-row-bound-proof} gives
\begin{equation}
    L_C
    \le
    2
    \max\left\{
        \sum_{\ell=0}^{S_D} (4\Lambda)^{\ell},
        1+4\Lambda\sum_{\ell=0}^{S_D}(4\Lambda)^{\ell}
    \right\}.
    \label{eq:LC-unified-locality-bound}
\end{equation}
Likewise, Eq.~\eqref{eq:Gtilde-inverse-row-sparsity-proof} gives
\begin{equation}
    r_0
    \le
    \max\left\{
        \sum_{\ell=0}^{S_D}\lceil\Lambda\rceil^\ell,
        1+\lceil\Lambda\rceil
        \sum_{\ell=0}^{S_D}\lceil\Lambda\rceil^\ell
    \right\}.
    \label{eq:response-inverse-sparsity-unified}
\end{equation}
Both bounds are uniform in $N,M$ under
Eq.~\eqref{eq:uniform-response-overlap-assumption}.

Since an $M$-row matrix with at most $r_0$ nonzero entries per row has at most
$Mr_0$ nonzero entries, $\operatorname{nnz}(\mathbf H)=\mathcal{O}(M)$.
\end{proof}


\subsection{Taylor bounds for normalized responses}
\label{app:response-taylor}

Let $m\ge1$ denote the Taylor truncation degree.
Each raw-response functional has total scalar coefficient weight at most one.
In both locality settings, the scalar matrix elements obey
\begin{equation}
    \norm{(\cL_\theta^\dagger)^k(O)}_\infty
    \le
    \Lambda^k k!,
    \label{eq:app-common-taylor}
\end{equation}
with the unified scale $\Lambda$ from
Eq.~\eqref{eq:unified-locality-scale}.  Combining this bound with
$\|\mathbf H_{\alpha\cdot}\|_1\le L_C$ gives
\begin{equation}
    \abs{\partial_t^k g_\alpha(0,\theta)}
    \le
    L_C\Lambda^k k!.
    \label{eq:app-normalized-taylor}
\end{equation}
When $\Lambda t<1$, Eq.~\eqref{eq:app-normalized-taylor} gives
\begin{equation}
    \left|
    g_\alpha(t,\theta)-\sum_{k=1}^m
    \frac{t^k}{k!}\partial_t^k g_\alpha(0,\theta)
    \right|
    \le
    L_C\frac{(\Lambda t)^{m+1}}{1-\Lambda t}.
    \label{eq:app-response-tail}
\end{equation}

\subsection{Classical construction of the truncated response polynomials}
\label{subsec:classical-coefficient-construction}

The single-time method uses coefficient tables for the truncated response map.
We bound their construction separately from the experimental response cost;
the square response family contains $M$ coefficients and $M$ raw responses.

\paragraph{Computational model.}
All complexity bounds in this section count sparse-table arithmetic
operations.  Table insertion, lookup, and symbolic Pauli-key operations are
unit cost.  A bit-complexity analysis would additionally track coefficient
precision, key representation, and the sparse input length
$\mathcal S_{\rm dict}$; we do not pursue that analysis here.

The hybrid family of real raw responses \(\{\cF_j:j\in\cJ_\cA\}\) has the
following scalar Pauli-matrix-element representation for each of its elements:
\begin{equation}
    \cF_j(\Phi)
    =
    \sum_{\nu\in\mathcal I_j}w_{j\nu}
    \frac1D\Tr[Q_{j\nu}\Phi(P_{j\nu})],
    \qquad
    \sum_{\nu\in\mathcal I_j}|w_{j\nu}|\le1.
    \label{eq:generic-raw-matrix-elements}
\end{equation}
For each \(j\), let \(a(j)\) be the candidate index whose coefficient is
paired with the \(j\)th raw-response row, and set
\begin{equation}
    W_j:=\supp(\cL_{a(j)}^\dagger).
    \label{eq:raw-response-associated-region}
\end{equation}
By construction, \(W_j\) contains the supports of every \(P_{j\nu}\) and
\(Q_{j\nu}\) in this representation, and \(|W_j|\le s\).
The number of scalar matrix elements satisfies
\begin{equation}
    |\mathcal I_j|
    \le
    \begin{cases}
        4^{|W_j|}, & \text{diagonal response},\\
        2\,4^{|W_j|}, & \text{short Hamiltonian or off-diagonal response},\\
        1, & \text{direct response},
    \end{cases}
    \label{eq:raw-scalar-element-counts}
\end{equation}
where the direct-response weight is $w_{j\nu}=\chi_C/2$.  Every averaged row has
$|W_j|\le S_D$, whereas a long-support Hamiltonian row is direct and contains
only one scalar matrix element.  Consequently,
\begin{equation}
    \sum_{j=1}^M|\mathcal I_j|
    \le
    2M4^{S_D}.
    \label{eq:total-scalar-raw-elements}
\end{equation}
For overlap counting under bounded dual-interaction-graph degree, every
candidate term whose support intersects \(W_j\) is either \(a(j)\) itself
or a neighbor of \(a(j)\) in the dual interaction graph.  Hence
\begin{equation}
    \left|
    \left\{
        a:
        \operatorname{supp}(\mathcal L_a^\dagger)\cap W_j\ne\emptyset
    \right\}
    \right|
    \le d+1 .
    \label{eq:candidate-term-degree-bound}
\end{equation}

For any Pauli product $Q$, write it uniquely as
$Q=\omega P$ with $\omega\in\{\pm1,\pm i\}$ and
$P\in\{I,X,Y,Z\}^{\otimes N}$.  We define
$
    \operatorname{pf}(Q):=P,
    \label{eq:phase-free-representative-definition}
$
so $\operatorname{pf}$ removes the global Pauli phase and retains only the
phase-free Pauli-string label.

\begin{lemma}[Deterministic phase-free Pauli transfer]
\label{lem:deterministic-pauli-transfer}
For every elementary real Pauli--GKSL term $\mathcal L_a^\dagger$ there
is a phase-free Pauli multiplier $R_a$ such that, for every phase-free Pauli
$P$,
\begin{equation}
    \mathcal L_a^\dagger(P)
    =
    \gamma_a(P)\,\operatorname{pf}(R_aP)
    \label{eq:deterministic-pauli-transfer}
\end{equation}
for a scalar $\gamma_a(P)$, possibly zero.  One may take $R_a=C$ for a
Hamiltonian term generated by $C$, $R_a=I$ for a diagonal dissipative
term, and $R_a=\operatorname{pf}(BA)$ for either off-diagonal term
generated by $A,B$.  Hence, for a fixed input $P$ and exponent vector
$\mathbf n$, every nonzero ordering of the corresponding terms has the
same phase-free output label
$\operatorname{pf}(P\prod_aR_a^{n_a})$.
\end{lemma}

\begin{proof}
The two terms of $i[C,P]$ are phase-equivalent to $CP$.  For a diagonal
term, $APA-P$ is a scalar multiple of $P$.  In
$\mathcal D_{A,B}^\dagger(P)$, the three products $BPA$, $BAP$, and $PBA$ are
all phase-equivalent to $BAP$.  Likewise, every term in
$\mathcal D_{B,A}^\dagger(P)$ is phase-equivalent to $ABP$, hence also to
$BAP$.  The same conclusion therefore holds for their real and imaginary
combinations, giving Eq.~\eqref{eq:deterministic-pauli-transfer}.
Under composition the multipliers accumulate, and Pauli strings commute up to
a phase.  Removing that phase leaves a label depending only on the
multiplicities $\mathbf n$, not on their ordering.
\end{proof}

For \(x\in\R^M\), define the \(k\)-th homogeneous raw response polynomial by
\begin{equation}
    f_{j,k}(x)
    :=
    \frac{1}{k!}
    \cF_j\!\left((\mathcal L_x^\dagger)^k\right).
    \label{eq:raw-homogeneous-polynomial}
\end{equation}
For real \(x\), the map \((\mathcal L_x^\dagger)^k\) is Hermiticity preserving,
so \(f_{j,k}(x)\in\R\).

\begin{lemma}[Compressed-monomial count under bounded dual-interaction-graph degree]
\label{lem:finite-degree-compressed-monomials}
Under the bounded dual-interaction-graph degree hypothesis, let $\cM_{j,k}$ be the set of
degree-$k$ commutative monomials with nonzero coefficient in the raw response
polynomial $f_{j,k}$.  Then, for every $j$ and $k\ge1$,
\begin{equation}
    |\cM_{j,k}|
    \le
    [5(d+1)]^k.
    \label{eq:finite-degree-compressed-monomial-count}
\end{equation}
Thus one may take $c_{\rm mon}^{\rm fd}=5$.
\end{lemma}

\begin{proof}
Adjoin to the candidate overlap graph a root vertex representing $W_j$, joined
to every candidate term whose support intersects $W_j$.  By
Eq.~\eqref{eq:candidate-term-degree-bound}, the augmented graph has maximum
degree at most $\Delta_{\rm fd}:=d+1$.  If a degree-$k$ ordered transfer word is nonzero,
each newly appearing term overlaps $W_j$ or a previously appearing
term.  Hence its $\ell$ distinct terms and the root form a connected
set.

A connected set with $\ell$ nonroot vertices has a rooted spanning tree with
$\ell$ edges.  There are at most $4^\ell$ rooted plane-tree shapes and at most
$\Delta_{\rm fd}^\ell$ choices of neighbor labels, so there are at most
$(4\Delta_{\rm fd})^\ell$ such connected sets.  For a fixed set of $\ell$ terms,
the positive multiplicities summing to $k$ can be chosen in
$\binom{k-1}{\ell-1}$ ways.  Therefore
\begin{align}
    |\cM_{j,k}|
    &\le
    \sum_{\ell=1}^k
    (4\Delta_{\rm fd})^\ell
    \binom{k-1}{\ell-1} \\
    &=
    4\Delta_{\rm fd}(1+4\Delta_{\rm fd})^{k-1}
    \le
    (5\Delta_{\rm fd})^k,
\end{align}
where the last inequality uses $\Delta_{\rm fd}\ge1$.  Merging different orderings can
only reduce the number of stored commutative monomials, proving the claim.
\end{proof}

The normalized homogeneous polynomial used by the exact compressed evaluator is
\begin{equation}
    (p_k(x))_\alpha
    :=\sum_j(\mathbf{H})_{\alpha j}f_{j,k}(x)
    =\frac1{k!}\mathsf C_\alpha\!\left((\cL_x^\dagger)^k\right).
    \label{eq:app-real-pk-definition}
\end{equation}
Expanding the definition gives
\begin{align}
    f_{j,k}(x)
    &=
    \frac{1}{k!}
    \sum_{\nu\in\mathcal I_j}w_{j\nu}
    \sum_{a_1,\ldots,a_k}
    x_{a_1}\cdots x_{a_k}
    \frac{1}{D}
    \operatorname{Tr}\!\left[
        Q_{j\nu}
        \mathcal L_{a_k}^\dagger\cdots
        \mathcal L_{a_1}^\dagger(P_{j\nu})
    \right].
    \label{eq:raw-homogeneous-expanded}
\end{align}
There are at most \(2\,4^{|W_j|}\) Pauli matrix elements in this representation
(and only one for a direct response).  The fixed factor two is absorbed into
the bookkeeping constants below.  It does not multiply the number of
distinct monomials in the stored polynomial table, because all input Pauli
strings have support contained in the same region \(W_j\) and hence involve
the same local variable set generated by the subsequent Pauli transfers.

By the definition of $r_0$ in Eq.~\eqref{eq:response-inverse-row-sparsity},
$\operatorname{nnz}(\mathbf H)\le Mr_0$.  In particular,
Theorem~\ref{thm:uniform-response-inverse} gives
$\operatorname{nnz}(\mathbf H)=\mathcal{O}(M)$ under either set of its bounded-overlap assumptions.

One nonzero local Pauli-transfer step has constant cost.  Thus an
ordered transfer word of length \(k\) costs \(\mathcal{O}(k)\) arithmetic operations up to
fixed Pauli-phase bookkeeping constants.

\begin{proposition}[Ordered coefficient construction under bounded dual-interaction-graph degree]
\label{prop:coef-construction-finite-degree}
Assume bounded dual-interaction-graph degree with maximum degree $d$ and let
$m\ge1$.  Ordered Pauli-transfer enumeration constructs the raw coefficient
tables through degree $m$ in time
\begin{equation}
    T_{\rm coef}^{\rm fd}(m)
    =
    M\exp\!\left(
        \mathcal{O}\!\left(
            S_D+m\log\!\bigl(m(1+\Lambda)\bigr)
        \right)
    \right).
    \label{eq:Tcoef-fd}
\end{equation}
\end{proposition}

\begin{proof}
For a fixed real raw response $\cF_j$, every input Pauli string $P_{j\nu}$ has
support contained in $W_j$.  By
Eq.~\eqref{eq:candidate-term-degree-bound}, the first nontrivial local
term has at most $d+1$ choices.  After $\ell$ terms have been
chosen, the next term has at most $(\ell+1)(d+1)$ choices.  Hence the
number of effective words of length $k$ is at most $(d+1)^k k!$.  Each word
costs $\mathcal{O}(k)$, and Eq.~\eqref{eq:total-scalar-raw-elements} bounds the total
number of scalar matrix elements by $2M4^{S_D}$.  Since
\begin{equation}
    \sum_{k=1}^m k(d+1)^k k!
    \le 2m(d+1)^m m!,
\end{equation}
using $m!\le m^m$, $\Lambda=2(d+1)$, and
$4^{S_D}=\exp(\mathcal{O}(S_D))$ proves Eq.~\eqref{eq:Tcoef-fd}.
\end{proof}

\begin{corollary}[Compressed storage under bounded dual-interaction-graph degree]
\label{cor:coef-storage-finite-degree}
After merging words that yield the same commutative monomial, the numbers of
stored raw- and normalized-response terms satisfy
\begin{align}
    N_{\rm poly,raw}^{\rm fd}(m)
    &=
    M\exp\!\left(\mathcal{O}\!\left(m\log(1+\Lambda)\right)\right),
    \label{eq:Npoly-raw-fd}\\
    N_{\rm poly,norm}^{\rm fd}(m)
    &=
    \operatorname{nnz}(\mathbf H)
    \exp\!\left(\mathcal{O}\!\left(m\log(1+\Lambda)\right)\right).
    \label{eq:Npoly-norm-fd}
\end{align}
\end{corollary}

\begin{proof}
Lemma~\ref{lem:finite-degree-compressed-monomials} gives at most
$[5(d+1)]^k$ degree-$k$ monomials per raw response.  Summing this geometric
bound over $k\le m$ proves Eq.~\eqref{eq:Npoly-raw-fd}; combining raw tables
through the nonzero entries of $\mathbf H$ proves
Eq.~\eqref{eq:Npoly-norm-fd}.
\end{proof}

\begin{proposition}[Ordered coefficient construction under unweighted local strength]
\label{prop:coef-construction-local-strength}
Assume unweighted local strength $\kappa_0$, maximum support size $s$, and
$m\ge1$.  Ordered Pauli-transfer enumeration constructs the raw coefficient
tables through degree $m$ in time
\begin{equation}
    T_{\rm coef}^{\rm ls}(m)
    =
    M\exp\!\left(
        \mathcal{O}\!\left(
            S_D+m\log\!\bigl(m(1+\Lambda)\bigr)
        \right)
    \right).
    \label{eq:Tcoef-ls}
\end{equation}
\end{proposition}

\begin{proof}
For a generated support $V\subseteq[N]$, Lemma~\ref{lem:normalized-term-norms}
and the definition of $\kappa_0$ give
\begin{equation}
    \left|\left\{a:
    \supp(\mathcal L_a^\dagger)\cap V\ne\emptyset\right\}\right|
    \le
    \sum_{y\in V}
    \sum_{a:y\in\supp(\mathcal L_a^\dagger)}
    \|\mathcal L_a^\dagger\|_{\infty\to\infty}
    \le \kappa_0|V|.
    \label{eq:ls-count-coef-construction}
\end{equation}
After $\ell$ transfers, $|V|\le(\ell+1)s$, so the number of effective
length-$k$ words is at most $\Lambda^k k!$.  Multiplying by the $\mathcal{O}(k)$ word
cost and the $\mathcal{O}(M4^{S_D})$ scalar matrix elements gives the degree sum.  Since
$\Lambda\ge1$ for a nonempty normalized dictionary,
\begin{equation}
    \sum_{k=1}^m k\Lambda^k k!\le2m\Lambda^m m!,
\end{equation}
which proves Eq.~\eqref{eq:Tcoef-ls}.
\end{proof}

\begin{corollary}[Compressed storage under unweighted local strength]
\label{cor:coef-storage-local-strength}
After merging words that yield the same commutative monomial,
\begin{align}
    N_{\rm poly,raw}^{\rm ls}(m)
    &=
    M\exp\!\left(\mathcal{O}\!\left(m\log(1+\Lambda)\right)\right),
    \label{eq:Npoly-raw-ls}\\
    N_{\rm poly,norm}^{\rm ls}(m)
    &=
    \operatorname{nnz}(\mathbf H)
    \exp\!\left(\mathcal{O}\!\left(m\log(1+\Lambda)\right)\right).
    \label{eq:Npoly-norm-ls}
\end{align}
\end{corollary}

\begin{proof}
For a candidate term \(a\), write
\begin{equation}
    X_a:=\supp(\mathcal L_a^\dagger).
\end{equation}
Lemma~\ref{lem:normalized-term-norms} implies that every candidate term
has induced norm at least one.  Hence
\begin{align}
    \left|
        \left\{
            b:X_b\cap X_a\ne\emptyset
        \right\}
    \right|
    &\le
    \sum_{\substack{b:\\X_b\cap X_a\ne\emptyset}}
    \|\mathcal L_b^\dagger\|_{\infty\to\infty}
    \notag\\
    &\le
    \sum_{y\in X_a}
    \sum_{b:y\in X_b}
    \|\mathcal L_b^\dagger\|_{\infty\to\infty}
    \notag\\
    &\le
    \kappa_0|X_a|
    \le
    s\kappa_0
    =
    \Lambda.
    \label{eq:local-strength-candidate-overlap-degree}
\end{align}

Attach a virtual root representing \(W_j\).  Since
\(|W_j|\le s\), the same argument gives
\begin{equation}
    \left|
        \left\{
            a:X_a\cap W_j\ne\emptyset
        \right\}
    \right|
    \le
    \kappa_0|W_j|
    \le
    \Lambda.
    \label{eq:local-strength-root-degree}
\end{equation}
For a candidate vertex, the set counted in
Eq.~\eqref{eq:local-strength-candidate-overlap-degree} includes the
vertex itself.  Removing the self-count and, when applicable, adding
the edge to the virtual root shows that the augmented graph has
maximum degree at most
$\Delta_{\rm ls}
    :=
    \max\{1,\lceil\Lambda\rceil\}.$
The rooted connected-set argument of
Lemma~\ref{lem:finite-degree-compressed-monomials} therefore gives at
most$(5\Delta_{\rm ls})^k$ degree-\(k\) monomials per raw response.  For a nonempty normalized
dictionary, \(\Lambda\ge1\), so $\Delta_{\rm ls}\le2\Delta$. Consequently,
\begin{equation}
    |\mathcal M_{j,k}|
    \le
    (10\Lambda)^k.
    \label{eq:local-strength-compressed-monomial-count}
\end{equation}
Summing this geometric bound over \(k\le m\) proves
Eq.~\eqref{eq:Npoly-raw-ls}.  Combining the raw tables through the
nonzero entries of \(\mathbf H\) proves
Eq.~\eqref{eq:Npoly-norm-ls}.
\end{proof}

For the compressed construction, set
$
    \Delta_\star
    :=
    \max\left\{1,\left\lceil\Lambda\right\rceil\right\}.
    \label{eq:compressed-dp-overlap-degree}
$
We use the sparse-table arithmetic model specified above, with expected
constant-time insertion and lookup and unit-cost symbolic Pauli-key operations.

\begin{proposition}[Compressed dynamic programming for response polynomials]
\label{prop:compressed-response-dp}
Assume either locality hypothesis in Eq.~\eqref{eq:unified-locality-scale} and
let $m\ge1$.  Sparse tables keyed by an output phase-free Pauli string and a
commutative coefficient monomial construct the raw response polynomials and
their normalized combinations in
\begin{equation}
    T_{\rm coef}^{\rm DP}(m)
    =
    \mathcal{O}\!\left(
        M4^{S_D}\Delta_\star
        \sum_{k=0}^{m-1}(k+1)(5\Delta_\star)^k
        +
        \operatorname{nnz}(\mathbf H)
        \sum_{k=1}^{m}(5\Delta_\star)^k
    \right)
    \label{eq:compressed-dp-complexity}
\end{equation}
arithmetic operations.
\end{proposition}

\begin{proof}
\emph{State count.}
For each scalar Pauli matrix element in
Eq.~\eqref{eq:generic-raw-matrix-elements}, write $P=P_{j\nu}$ and define
\begin{equation}
    V_{\nu,k}(x)
    :=
    \frac1{k!}(\mathcal L_x^\dagger)^k(P).
    \label{eq:compressed-dp-Vk}
\end{equation}
Expand $V_{\nu,k}$ in a sparse table
\begin{equation}
    V_{\nu,k}(x)
    =
    \sum_{|\mathbf n|=k}
    \sum_{R\in\mathcal P_{[N]}}
    A_{\nu,k}(R,\mathbf n)x^{\mathbf n}R,
    \label{eq:compressed-dp-table}
\end{equation}
where $R$ is phase free and all phases are stored in
$A_{\nu,k}(R,\mathbf n)$.  Initialize with $(P,\mathbf0)$ and use
\begin{equation}
    V_{\nu,k+1}(x)
    =
    \frac1{k+1}\mathcal L_x^\dagger(V_{\nu,k}(x)).
    \label{eq:compressed-dp-recurrence}
\end{equation}
Every transfer increments one exponent, and equal pairs $(R,\mathbf n)$ are
merged immediately.  Lemma~\ref{lem:deterministic-pauli-transfer} gives at
most one phase-free output label per transition.  Pauli orthogonality at the
final step selects $R=Q_{j\nu}$ and produces $f_{j,k}$.

Adjoining the raw region as a root gives augmented degree at most
$\Delta_\star$ under either locality condition.  The rooted plane-tree
argument of Lemma~\ref{lem:finite-degree-compressed-monomials} bounds the
number of degree-$k$ monomials by $(5\Delta_\star)^k$.  For a fixed scalar
input and monomial there is at most one phase-free output key.  Since an
averaged row has at most $2\,4^{S_D}$ scalar inputs, its degree-$k$ table has
at most
$
    2\,4^{S_D}(5\Delta_\star)^k
    \label{eq:compressed-dp-state-count}
$
states.  Each state has at most $(k+1)\Delta_\star$ outgoing transitions.
Summing transitions over $k=0,\ldots,m-1$ and all $M$ rows gives the first
term of Eq.~\eqref{eq:compressed-dp-complexity}; applying the nonzero entries
of $\mathbf H$ gives the second.
\end{proof}

\begin{corollary}[Polylogarithmic compressed construction]
\label{cor:compressed-response-dp-polylog}
If $S_D,\Delta_\star=\mathcal{O}(1)$ and
$\operatorname{nnz}(\mathbf H)=\mathcal{O}(M)$, then $m=\mathcal{O}(\log\log M)$ is a sufficient condition for $T_{\rm coef}^{\rm DP}(m)
    =M\operatorname{polylog}(M)$.
\end{corollary}

\begin{proof}
For fixed $S_D$ and $\Delta_\star$, both geometric sums in
Eq.~\eqref{eq:compressed-dp-complexity}, including their polynomial-in-$m$
prefactors, are $\operatorname{polylog}(M)$ when
$m=\mathcal{O}(\log\log M)$.
\end{proof}

\section{Chebyshev--Lobatto response interpolation}
\label{app:chebyshev}

We first derive the interpolation bound for a scalar response $g$;
Algorithm~\ref{alg:cheb} applies it coordinatewise.

Let $r\ge2$ and choose the shifted Chebyshev--Lobatto nodes
\begin{equation}
    t_j=\frac{\tau}{2}\left(1-\cos\frac{j\pi}{r}\right),
    \qquad
    j=0,1,\ldots,r.
    \label{eq:app-lobatto-nodes}
\end{equation}
Then $t_0=0$ and $t_r=\tau$.  Let $\ell_j$ be the Lagrange basis polynomial
satisfying $\ell_j(t_k)=\delta_{jk}$, and define the endpoint derivative
weights
\begin{equation}
    w_j=\ell_j'(0).
    \label{eq:app-derivative-weights}
\end{equation}
For nodal data $y=(y_0,\ldots,y_r)$, define the endpoint derivative
functional
\begin{equation}
    \mathsf D_r[y]:=\sum_{j=0}^r w_jy_j.
    \label{eq:app-endpoint-derivative-functional}
\end{equation}
This equals the derivative at the origin of the degree-$r$ interpolant of the
data.

\begin{lemma}[Endpoint derivative-weight bound]
\label{lem:app-weight-bound}
For the nodes in Eq.~\eqref{eq:app-lobatto-nodes}, the derivative weights obey
\begin{equation}
    W_r:=\sum_{j=0}^r\abs{w_j}
    \le
    \frac{3r^2}{\tau}.
    \label{eq:app-Wm-bound}
\end{equation}
\end{lemma}

\begin{proof}
We map the interval $[0,\tau]$ to $[-1,1]$ by $x=1-2t/\tau$.  The node $t=0$
corresponds to $x=1$.  For Chebyshev--Lobatto nodes
$x_j=\cos(j\pi/r)$, the first row of the standard differentiation matrix
\(\mathsf D\) is
\begin{equation}
    \mathsf D_{00}=\frac{2r^2+1}{6},
    \qquad
    \mathsf D_{0j}=\frac{2(-1)^j}{c_j(1-x_j)}
    \quad (1\le j\le r),
    \label{eq:app-cheb-diff-row}
\end{equation}
where $c_j=1$ for $1\le j\le r-1$ and $c_r=2$.  Differentiation with respect to
$t$ multiplies the derivative with respect to $x$ by $-2/\tau$, so
\begin{equation}
    W_r
    \le
    \frac2\tau\left[
        \frac{2r^2+1}{6}
        +
        2\sum_{j=1}^{r-1}\frac1{1-\cos(j\pi/r)}
        +
        \frac12
    \right].
\end{equation}
Using $1-
\cos u=2\sin^2(u/2)$ and the identity
\begin{equation}
    \sum_{j=1}^{r-1}\cscop^2\left(\frac{j\pi}{2r}\right)
    =
    \frac{2(r^2-1)}{3},
\end{equation}
we obtain $W_r\le3r^2/\tau$ for all $r\ge2$.
\end{proof}

\begin{theorem}[Derivative estimation from noisy Lobatto data]
\label{thm:app-cheb-error}
Assume
\begin{equation}
    g(t)=\sum_{k\ge1}a_kt^k,
    \qquad
    |a_k|\le L_C\Lambda^k,
    \label{eq:app-g-series-bound}
\end{equation}
for $0\le t\le\tau$, with $\Lambda\tau<1$.  If the nodal data satisfy
\begin{equation}
    \max_{0\le j\le r}|\widehat g(t_j)-g(t_j)|
    \le\xi_{\rm stat},
\end{equation}
let $p_r=I_rg$ be the exact-data interpolant, then
\begin{equation}
    \left| \widehat p_r'(0)-g'(0)\right|
    \le
    \frac{3r^2}{\tau}\xi_{\rm stat}
    +
    \frac{3L_Cr^2}{\tau}
    \frac{(\Lambda\tau)^{r+1}}{1-\Lambda\tau}.
    \label{eq:app-cheb-error-bound}
\end{equation}
\end{theorem}

\begin{proof}
Let
\begin{equation}
    T_r(t)=\sum_{k=1}^r a_k t^k
\end{equation}
be the degree-$r$ Taylor polynomial of $g$, and set $h(t)=g(t)-T_r(t)$.  The
interpolation operator $I_r$ is exact on polynomials of degree at most $r$. Since $I_r$ is exact on $T_r$,
\begin{equation}
    p_r=I_r g=I_rT_r+I_rh=T_r+I_rh.
\end{equation}
Since $h$ has no constant or linear term, $h'(0)=0$.  Hence the deterministic
derivative error is the derivative of the interpolated Taylor tail:
\begin{equation}
    p_r'(0)-g'(0)=(I_rh)'(0).
\end{equation}
By the definition of the derivative weights,
\begin{equation}
    \abs{(I_rh)'(0)}
    \le
    \sum_{j=0}^r\abs{w_j}\abs{h(t_j)}.
\end{equation}
For every node $t_j\in[0,\tau]$, the coefficient bound gives
\begin{equation}
    \abs{h(t_j)}
    \le
    \sum_{k=r+1}^\infty L_C(\Lambda t_j)^k
    \le
    L_C\frac{(\Lambda\tau)^{r+1}}{1-\Lambda\tau}.
\end{equation}
Lemma~\ref{lem:app-weight-bound} then gives the deterministic term in
Eq.~\eqref{eq:app-cheb-error-bound}, namely the interpolation-bias term.

The polynomial $\widehat p_r-p_r$ has nodal values
$\widehat g(t_j)-g(t_j)$.  Its derivative at zero is
$\sum_j w_j(\widehat g(t_j)-g(t_j))$, so the endpoint weight bound gives
\begin{equation}
    \abs{\widehat p_r'(0)-p_r'(0)}
    \le
    W_r\xi_{\rm stat}
    \le
    \frac{3r^2}{\tau}\xi_{\rm stat}.
\end{equation}
Adding the deterministic and noisy bounds yields
Eq.~\eqref{eq:app-cheb-error-bound}.
\end{proof}

For $\Lambda>0$ and target accuracy $\epsilon>0$, define the canonical
parameters
\begin{align}
    \tau_\star
    &:=
    \frac{1}{2\Lambda},
    \label{eq:canonical-tau-choice}\\
    r_\star
    &:=
    \min\left\{
    n\ge2:
    2^n\ge
    \frac{12L_C\Lambda}{\epsilon}n^2
    \right\},
    \label{eq:canonical-degree-choice}\\
    \xi_\star
    &:=
    \frac{\epsilon}{12\Lambda r_\star^2}.
    \label{eq:canonical-stat-choice}
\end{align}

\begin{corollary}[Canonical parameters for derivative accuracy]
\label{thm:sufficient-derivative-accuracy}
Under the assumptions of Theorem~\ref{thm:app-cheb-error}, take
$\tau=\tau_\star$, $r=r_\star$, and
$\xi_{\rm stat}\le\xi_\star$.  Then
\begin{equation}
    \left|\widehat p_r'(0)-g'(0)\right|
    \le\epsilon.
\end{equation}
\end{corollary}

\begin{proof}
With $\tau=1/(2\Lambda)$, Theorem~\ref{thm:app-cheb-error} bounds the bias by
$6L_C\Lambda r^2/2^r$ and the noise by
$6\Lambda r^2\xi_{\rm stat}$.  The degree and accuracy choices make these
terms at most $\epsilon/2$ each.
\end{proof}
\section{Single-time projected response contraction}
\label{app:single-time}

At a fixed evolution time $t$, the method estimates
$g(t,\theta)\in\R^M$ and inverts its truncated response map.  Set
$\Omega=[-1,1]^M$.  We write
\begin{equation}
    g_m(t,\cdot):\Omega\longrightarrow\R^M,
    \qquad
    g_m(t,x)
    =
    \sum_{k=1}^m t^k p_k(x),
    \qquad
    p_k:\Omega\longrightarrow\R^M,
    \qquad p_1(x)=x,
    \label{eq:app-gm}
\end{equation}
where $p_k$ is defined in Eq.~\eqref{eq:app-real-pk-definition}.
The statistical estimate of the finite-time real response vector is
$\widehat g\in\R^M$.

Let $\Pi_\Omega$ denote coordinatewise clipping.
For $t>0$ and a response vector $y\in\R^M$, define
\begin{equation}
    \mathcal T_y(x)
    :=
    \Pi_\Omega\!\left[
        x-t^{-1}\bigl(g_m(t,x)-y\bigr)
    \right].
    \label{eq:projected-response-map}
\end{equation}
Given any $x^{(0)}\in\Omega$, write
\begin{equation}
    x^{(n+1)}=\mathcal T_y(x^{(n)}).
    \label{eq:projected-response-iteration}
\end{equation}

\begin{theorem}[Global convergence of projected response contraction]
\label{thm:app-projected-contraction}
Assume
\begin{equation}
    \sup_{x\in\Omega}
    \left\|I_M-t^{-1}D_xg_m(t,x)\right\|_{\infty\to\infty}
    \le q<1.
    \label{eq:projected-contraction-jacobian}
\end{equation}
Then, for every $y\in\R^M$, $\mathcal T_y$ is a contraction on $\Omega$ with
factor $q$.  It has a unique fixed point $\bar x_y\in\Omega$, and the
iteration in Eq.~\eqref{eq:projected-response-iteration} converges to
$\bar x_y$ from every initialization in $\Omega$.
\end{theorem}

\begin{proof}
Coordinatewise clipping is nonexpansive in the $\ell_\infty$ norm:
\begin{equation}
    \|\Pi_\Omega(u)-\Pi_\Omega(v)\|_\infty
    \le
    \|u-v\|_\infty.
    \label{eq:box-projection-nonexpansive}
\end{equation}
Define the unprojected map
\begin{equation}
    \mathcal S_y(x)
    :=
    x-t^{-1}\bigl(g_m(t,x)-y\bigr).
\end{equation}
For $x,z\in\Omega$, convexity of $\Omega$ and the fundamental theorem of
calculus give
\begin{equation}
    \mathcal S_y(x)-\mathcal S_y(z)
    =
    \int_0^1
    \left[I_M-t^{-1}D_xg_m(t,z+u(x-z))\right](x-z)\,du.
\end{equation}
Equation~\eqref{eq:projected-contraction-jacobian} and projection
nonexpansiveness imply
\begin{equation}
    \|\mathcal T_y(x)-\mathcal T_y(z)\|_\infty
    \le q\|x-z\|_\infty.
\end{equation}
The Banach fixed-point theorem gives the claimed fixed point and convergence.
\end{proof}

\begin{corollary}[Perturbation bound for projected response contraction]
\label{cor:app-projected-perturbation}
Assume the hypotheses of Theorem~\ref{thm:app-projected-contraction} for some
$t>0$ and $0\le q<1$.  Let $y=\widehat g$, let $\theta\in\Omega$, and suppose
\begin{equation}
    \|g(t,\theta)-g_m(t,\theta)\|_\infty\le\eta_{\rm tail},
    \qquad
    \|\widehat g-g(t,\theta)\|_\infty\le\eta_{\rm stat}.
    \label{eq:projected-response-errors}
\end{equation}
With $\Delta:=\eta_{\rm tail}+\eta_{\rm stat}$, the iterates satisfy
\begin{equation}
    \|x^{(n)}-\theta\|_\infty
    \le
    q^n\|x^{(0)}-\theta\|_\infty
    +
    \frac{1-q^n}{t(1-q)}\Delta.
    \label{eq:projected-iteration-error}
\end{equation}
For $x^{(0)}=0$,
\begin{equation}
    \|x^{(n)}-\theta\|_\infty
    \le
    q^n+\frac{\Delta}{t(1-q)}.
    \label{eq:projected-zero-init-error}
\end{equation}
The fixed point $\bar x_{\widehat g}$ obeys
\begin{equation}
    \|\bar x_{\widehat g}-\theta\|_\infty
    \le
    \frac{\eta_{\rm tail}+\eta_{\rm stat}}{t(1-q)}.
    \label{eq:projected-fixed-point-error}
\end{equation}
\end{corollary}

\begin{proof}
Because $\theta\in\Omega$, $\Pi_\Omega(\theta)=\theta$.  Applying projection
nonexpansiveness and Theorem~\ref{thm:app-projected-contraction} gives
\begin{align}
    \|x^{(n+1)}-\theta\|_\infty
    &\le
    q\|x^{(n)}-\theta\|_\infty
    +\frac1t\|\widehat g-g_m(t,\theta)\|_\infty\\
    &\le
    q\|x^{(n)}-\theta\|_\infty+\frac\Delta t.
    \label{eq:projected-one-step-recurrence}
\end{align}
Iterating the recurrence proves Eq.~\eqref{eq:projected-iteration-error}.
For $x^{(0)}=0$, $\|x^{(0)}-\theta\|_\infty\le1$, which gives
Eq.~\eqref{eq:projected-zero-init-error}.  Taking $n\to\infty$ gives
Eq.~\eqref{eq:projected-fixed-point-error}.
\end{proof}

\begin{remark}[Exact-data case]
If $\eta_{\rm tail}=\eta_{\rm stat}=0$, then $\bar x_{\widehat g}=\theta$, and the
zero-initialized iteration converges geometrically to the true parameter.
\end{remark}

\subsection{Sufficient Jacobian conditions for the finite-time response map}

Recall $g_m$ from Eq.~\eqref{eq:app-gm}.  The required Jacobian condition is
\begin{equation}
    \sup_{x\in\Omega}
    \left\|I_M-t^{-1}D_x g_m(t,x)\right\|_{\infty\to\infty}
    \le q<1 .
    \label{eq:main-jacobian-stability}
\end{equation}
The next theorem gives a direct sufficient bound for its left-hand side.

\begin{theorem}[Unified Jacobian remainder bound]
\label{thm:unified-jacobian}
Assume either locality hypothesis in Eq.~\eqref{eq:unified-locality-scale} and
let $\Omega\subseteq[-1,1]^M$ be convex.  Then, for every $t>0$,
\begin{equation}
    \sup_{x\in\Omega}
    \left\|I_M-t^{-1}D_xg_m(t,x)\right\|_{\infty\to\infty}
    \le
    L_C\Lambda
    \sum_{k=2}^m k(\Lambda t)^{k-1}.
    \label{eq:unified-jacobian-finite-sum}
\end{equation}
\end{theorem}

\begin{proof}
Since $g_m(t,x)=\sum_{k=1}^m t^kp_k(x)$ and $p_1(x)=x$,
\begin{equation}
    I_M-t^{-1}D_xg_m(t,x)
    =
    -\sum_{k=2}^m t^{k-1}Dp_k(x).
\end{equation}
Fix $x\in\Omega$ and $\|v\|_\infty\le1$.  From
Eq.~\eqref{eq:app-real-pk-definition},
\begin{equation}
    (Dp_k(x)v)_\alpha
    =
    \frac1{k!}
    \sum_{\ell=0}^{k-1}
    \mathsf C_\alpha\!\left(
        (\cL_x^\dagger)^\ell
        \cL_v^\dagger
        (\cL_x^\dagger)^{k-1-\ell}
    \right).
\end{equation}
In the bounded-degree setting,
Theorem~\ref{thm:app-finite-degree} applies directly to the sequence of
coefficient vectors consisting of \(k-1\) copies of \(x\) and one copy
of \(v\).  It therefore bounds every one-insertion product by
\(\Lambda^k k!\).

For the local-strength setting, fix an insertion position and write the
corresponding product as
\begin{equation}
    \mathcal L_{z^{(k)}}^\dagger
    \cdots
    \mathcal L_{z^{(1)}}^\dagger,
    \qquad
    z^{(r)}\in\{x,v\}.
\end{equation}
Every coefficient vector in this sequence satisfies
\(\|z^{(r)}\|_\infty\le1\).  For a raw-response input Pauli \(O\), we
have \(|\supp(O)|\le s\).  After \(r-1\) elementary transfers, the union
of the initial support and the supports of the selected terms has size
at most \(rs\).  As in
Eq.~\eqref{eq:local-strength-power-one-step}, the total induced-norm
weight of all terms that can contribute at the \(r\)-th step is
therefore at most $rs\kappa_0=r\Lambda.$ 
Multiplying these bounds for \(r=1,\ldots,k\) gives
\begin{equation}
    \left\|
        \mathcal L_{z^{(k)}}^\dagger
        \cdots
        \mathcal L_{z^{(1)}}^\dagger(O)
    \right\|_\infty
    \le
    \prod_{r=1}^k r\Lambda
    =
    k!\Lambda^k.
    \label{eq:local-strength-one-insertion-bound}
\end{equation}
Thus the same \(\Lambda^k k!\) bound holds for every insertion position
under either locality hypothesis.  Since
$\|\mathbf H_{\alpha\cdot}\|_1\le L_C$, summing over the $k$ insertion
positions gives
\begin{equation}
    \sup_{x\in\Omega}
    \|Dp_k(x)\|_{\infty\to\infty}
    \le
    L_C\Lambda\,k\Lambda^{k-1}.
    \label{eq:unified-Dpk-bound}
\end{equation}
Substitution proves Eq.~\eqref{eq:unified-jacobian-finite-sum}.
\end{proof}

\begin{corollary}[A simple sufficient small-time condition]
\label{cor:unified-jacobian-small-time}
Under the hypotheses of Theorem~\ref{thm:unified-jacobian}, let $0<q<1$ and
suppose $L_C\Lambda>0$.  If
\begin{equation}
    0<t
    <
    \frac1\Lambda
    \min\left\{\frac12,\frac{q}{6L_C\Lambda}\right\},
    \label{eq:unified-jacobian-simple-time}
\end{equation}
then
\begin{equation}
    \sup_{x\in\Omega}
    \left\|I_M-t^{-1}D_xg_m(t,x)\right\|_{\infty\to\infty}
    <q.
    \label{eq:unified-jacobian-conclusion}
\end{equation}
\end{corollary}

\begin{proof}
For $\Lambda t<1$,
\begin{equation}
    \sum_{k=2}^{\infty}k(\Lambda t)^{k-1}
    =
    \frac{\Lambda t(2-\Lambda t)}{(1-\Lambda t)^2}.
\end{equation}
When $\Lambda t\le1/2$, the right-hand side is at most
$6\Lambda t$.  Theorem~\ref{thm:unified-jacobian} and
Eq.~\eqref{eq:unified-jacobian-simple-time} give the claim.
\end{proof}

Define the canonical contraction parameters by
\begin{equation}
    q_\star:=\frac12,
    \qquad
    t_\star:=\frac{1}{4\Lambda(1+6L_C\Lambda)}.
    \label{eq:canonical-projected-parameters}
\end{equation}

\begin{corollary}[Canonical iteration and truncation scales]
\label{cor:app-projected-iteration-count}
Assume either locality hypothesis in Eq.~\eqref{eq:unified-locality-scale},
let $\epsilon\in(0,1)$, and suppose $L_C,\Lambda=\mathcal{O}(1)$ uniformly in $N,M$
for a nonempty normalized dictionary.  Initialize the projected iteration at
$x^{(0)}=0$, use $q=q_\star$ and $t=t_\star$, and allocate
\begin{equation}
    \eta_{\rm tail},\eta_{\rm stat}
    \le
    \frac14t_\star(1-q_\star)\epsilon.
\end{equation}
Then the Taylor degree and iteration count can be chosen so that
\begin{equation}
    t_\star^{-1}=\mathcal{O}(\Lambda^2),
    \qquad
    m,n_{\rm it}=\mathcal{O}(\log(1/\epsilon)),
    \qquad
    \|x^{(n_{\rm it})}-\theta\|_\infty\le\epsilon.
    \label{eq:projected-corollary-result}
\end{equation}
\end{corollary}

\begin{proof}
For $q=1/2$, take
\begin{equation}
    n_{\rm it}
    \ge
    \left\lceil
    \frac{\log(2/\epsilon)}{-\log q}
    \right\rceil.
    \label{eq:projected-exact-iteration-count}
\end{equation}
The error allocation gives
\begin{equation}
    \frac{\eta_{\rm tail}+\eta_{\rm stat}}{t(1-q)}
    \le\frac{\epsilon}{2},
    \qquad
    q^{n_{\rm it}}\le\frac{\epsilon}{2}.
\end{equation}
Equation~\eqref{eq:projected-zero-init-error} therefore proves the accuracy
claim.  Set $a=\Lambda t<1$.  Equation~\eqref{eq:app-response-tail} meets the
tail allocation when
\begin{equation}
    m
    \ge
    \max\left\{
        1,\,
        \left\lceil
        \frac{
            \log\!\left(
                \dfrac{4L_C}
                {t(1-q)\epsilon(1-a)}
            \right)
        }{-\log a}
        \right\rceil-1
    \right\}.
    \label{eq:projected-truncation-degree}
\end{equation}
For the choice in the statement,
$\Lambda t_\star=1/[4(1+6L_C\Lambda)]<1/2$ and
$12L_C\Lambda^2t_\star=3L_C\Lambda/(1+6L_C\Lambda)<1$.
Hence Corollary~\ref{cor:unified-jacobian-small-time} gives $q=1/2$ for every finite $m$.
Moreover, Eq.~\eqref{eq:unified-locality-scale} and
Lemma~\ref{lem:normalized-term-norms} imply $\Lambda\ge1$, and
\begin{equation}
    t_\star^{-1}
    =4\Lambda(1+6L_C\Lambda)
    =\mathcal{O}(\Lambda^2),
    \qquad
    a_\star:=\Lambda t_\star
    =\frac{1}{4(1+6L_C\Lambda)}.
    \label{eq:canonical-projected-time-scaling}
\end{equation}
The quantity $a_\star$ is bounded away from zero and one uniformly in $N,M$.
Together with $L_C,\Lambda=\mathcal{O}(1)$,
Eqs.~\eqref{eq:projected-truncation-degree} and
\eqref{eq:projected-exact-iteration-count} prove
Eq.~\eqref{eq:projected-corollary-result}.
\end{proof}

\section{Algorithms and resource accounting}
\label{app:algorithms}

The algorithms below use $n_t$ for the number of measured time points.  The
time $t=0$ in
Chebyshev--Lobatto response interpolation is known exactly because
$g_\alpha(0,\theta)=0$, so it does not require state preparations.

For resource accounting, we separate shot aggregation from post-mean
processing.  Forming the empirical raw-response means requires $\mathcal{O}(N)$
arithmetic operations for $N$ measurement shots.  The detailed bounds below
begin after these means have been formed; one-time dictionary-representation
and preprocessing costs are reported separately.  Unless stated otherwise,
support-incidence lists are supplied with the dictionary, and the cost of
explicitly storing long Pauli strings is tracked through
$\mathcal S_{\rm long}$.

\begin{algorithm*}[H]
\caption{Chebyshev--Lobatto response interpolation}
\label{alg:cheb}
\KwIn{Accuracy $\epsilon$, failure probability $\delta$, Taylor constants
$L_C,\Lambda$, inverse $\mathbf{H}$ of the real response matrix, and the
family of real raw responses $\{\cF_j\}$}
\KwOut{Coefficient estimate $\widehat\theta\in\R^M$}

Set $\tau=\tau_\star$ and $r=r_\star$ from
Eqs.~\eqref{eq:canonical-tau-choice} and~\eqref{eq:canonical-degree-choice}.\\

Construct nodes
$t_\ell=\frac{\tau}{2}(1-\cos(\ell\pi/r))$ for
$\ell=0,1,\ldots,r$.\\

Compute the derivative weights $w_\ell$ associated with the Lagrange basis on
these nodes.\\

Set $W_r=\sum_{\ell=0}^r\abs{w_\ell}$ and
$\xi_{\rm stat}=\epsilon/(2W_r)$.\\

For each raw response $\cF_j$ and each nonzero node $t_\ell$, estimate
$\cF_j(e^{t_\ell\cL_\theta^\dagger})$ to accuracy
$\xi_{\rm stat}/L_C$ using the randomized Pauli measurement protocol.\\

For each coordinate $\alpha$ and each time $t_\ell$, form
$\widehat g_\alpha(t_\ell)=\sum_j(\mathbf{H})_{\alpha j}\widehat{\cF}_j(t_\ell)
-\mathsf C_\alpha(\Id)$.
Set $\widehat g_\alpha(0)=0$.\\

Set
$\widehat\theta_\alpha=\sum_{\ell=0}^{r}w_\ell\widehat g_\alpha(t_\ell)$ for
every $\alpha$.\\
\end{algorithm*}

\begin{proposition}[Sample complexity for Algorithm~\ref{alg:cheb}]
\label{prop:app-cheb-samples}
Algorithm~\ref{alg:cheb} uses
\begin{equation}
    N_{\rm Cheb}
    =
    \mathcal{O}\!\left(
        \frac{L_C^2M r^5}{\tau^2\epsilon^2}
        \log\frac{M r}{\delta}
    \right)
\end{equation}
state preparations and measurements and satisfies
$\norm{\widehat\theta-\theta}_\infty\le\epsilon$ with probability at least
$1-\delta$.
\end{proposition}

\begin{proof}
Theorem~\ref{thm:app-cheb-error} shows that the deterministic interpolation
bias is at most $\epsilon/2$ by the choice of $r$.  The same theorem
shows that the statistical contribution is at most $\epsilon/2$ whenever every
normalized response value at every measured node is estimated to accuracy
$\xi_{\rm stat}=\epsilon/(2W_r)$.  Proposition~\ref{prop:app-response-samples}
with $n_t=r$ and $\varepsilon_g=\xi_{\rm stat}$ gives
\begin{equation}
    N_{\rm Cheb}
    =
    \mathcal{O}\!\left(
        Mr
        \frac{L_C^2W_r^2}{\epsilon^2}
        \log\frac{M r}{\delta}
    \right).
\end{equation}
Lemma~\ref{lem:app-weight-bound} gives $W_r\le3r^2/\tau$, and the
stated bound follows.  The probability statement follows from the union bound
already included in Proposition~\ref{prop:app-response-samples}.
\end{proof}

After the empirical means have been formed, the classical cost is
\begin{equation}
    \mathcal{O}\!\left(r^2+Mr+r\operatorname{nnz}(\mathbf H)\right).
    \label{eq:app-cheb-postmean-cost}
\end{equation}
The $r^2$ term computes the interpolation weights.  The remaining terms
combine the raw responses and apply the derivative weights.  If
$\operatorname{nnz}(\mathbf H)=\mathcal{O}(M)$, the cost is linear in $M$ up to the
interpolation factor.
Under the uniform bounded-overlap condition,
$r=\operatorname{polylog}(1/\epsilon)$ and
$\operatorname{nnz}(\mathbf H)=\mathcal{O}(M)$, so the post-mean cost is
$M\operatorname{polylog}(1/\epsilon)$.  Including the
$\mathcal{O}(N_{\rm Cheb})$ shot-aggregation operations, the total classical
response-processing cost is $\widetilde{\mathcal{O}}(M/\epsilon^2)$.

\begin{algorithm}[H]
\caption{Single-time projected response contraction}
\label{alg:projected-contraction}
\KwIn{Accuracy $\epsilon$, failure probability $\delta$, time $t>0$ with
$\Lambda t<1$, Taylor
constants $L_C,\Lambda$, Jacobian contraction factor $0\le q<1$, inverse
$\mathbf{H}$ of the real response matrix, and the family of real raw responses
$\{\cF_j\}$}
\KwOut{Coefficient estimate $\widehat\theta\in[-1,1]^M$}

Set $\Omega=[-1,1]^M$ and
$\eta_{\rm stat}=t(1-q)\epsilon/4$.\\

Choose the smallest Taylor degree $m\ge1$ satisfying
\[
    L_C\frac{(\Lambda t)^{m+1}}{1-\Lambda t}
    \le
    \frac14t(1-q)\epsilon .
\]

Estimate all normalized responses at the single time $t$ so that
$\norm{\widehat g-g(t,\theta)}_\infty\le\eta_{\rm stat}$ with probability at
least $1-\delta$.\\

Precompute the coefficient tables for the truncated real response map
$g_m(t,x)$ on $\Omega$.\\

Set $x^{(0)}=0$ and
\[
    n_{\rm it}
    =
    \left\lceil
        \frac{\log(2/\epsilon)}{-\log q}
    \right\rceil
\]
for $0<q<1$; set $n_{\rm it}=1$ when $q=0$.\\

For $n=0,\ldots,n_{\rm it}-1$, update
\[
    x^{(n+1)}
    =
    \Pi_\Omega\!\left[
        x^{(n)}
        -
        t^{-1}\bigl(g_m(t,x^{(n)})-\widehat g\bigr)
    \right].
\]

Set $\widehat\theta=x^{(n_{\rm it})}$.\\
\end{algorithm}

\begin{proposition}[Sample count for Algorithm~\ref{alg:projected-contraction}]
\label{prop:app-projected-samples}
Let $0<\epsilon,\delta<1$.  Assume the hypotheses of
Theorem~\ref{thm:app-projected-contraction} with $t>0$ and $0\le q<1$, and
assume the normalized-response Taylor tail
Eq.~\eqref{eq:app-response-tail}.  Choose $m$, $\eta_{\rm stat}$, and
$n_{\rm it}$ as in Algorithm~\ref{alg:projected-contraction}.  Then its
response-estimation stage uses
\begin{equation}
    N_{\rm proj}
    =
    \mathcal{O}\!\left(
        \frac{L_C^2M}{t^2(1-q)^2\epsilon^2}
        \log\frac{M}{\delta}
    \right)
    \label{eq:app-projected-sample-count}
\end{equation}
state preparations and measurements and returns an estimate satisfying
$\|\widehat\theta-\theta\|_\infty\le\epsilon$ with probability at least
$1-\delta$.
\end{proposition}

\begin{proof}
This is Proposition~\ref{prop:app-response-samples} with $n_t=1$ and
$\varepsilon_g=\eta_{\rm stat}=t(1-q)\epsilon/4$.  The factor $M$ is present
because the method estimates all $M$ response coordinates at the chosen time.
The choice of $m$ makes
$\eta_{\rm tail}\le t(1-q)\epsilon/4$, while the chosen iteration count gives
$q^{n_{\rm it}}\le\epsilon/2$ (with $q^{n_{\rm it}}=0$ when $q=0$).
Corollary~\ref{cor:app-projected-perturbation} therefore gives the stated
accuracy.
\end{proof}

Let $T_{\rm coef}(m)$ denote the one-time coefficient-table construction
cost, and let $N_{\rm mon}(m)$ bound the stored monomials per raw response
through degree $m$.

\begin{proposition}[Master classical-cost bound for projected response contraction]
\label{prop:algorithm2-classical-cost}
After the empirical raw-response means have been formed, the total coefficient
construction and online iteration cost is
\begin{equation}
    T_{\rm proj}^{\rm cl}(m)
    =
    T_{\rm coef}(m)
    +
    \mathcal{O}\!\left(
        \operatorname{nnz}(\mathbf H)
        n_{\rm it}N_{\rm mon}(m)
        \operatorname{poly}(m)
    \right).
    \label{eq:alg2-classical-master-bound}
\end{equation}
\end{proposition}

\begin{proof}
The coefficient tables are constructed once.  Combining them through the
nonzero entries of $\mathbf H$ and evaluating the compressed tables costs
$\mathcal{O}(\operatorname{nnz}(\mathbf H)N_{\rm mon}(m)\operatorname{poly}(m))$ per
iteration.  The $\mathcal{O}(M)$ clipping cost is absorbed because
$\operatorname{nnz}(\mathbf H)\ge M$.  Multiplication by $n_{\rm it}$ proves
Eq.~\eqref{eq:alg2-classical-master-bound}.
\end{proof}

The ordered constructions are supplied by
Propositions~\ref{prop:coef-construction-finite-degree}
and~\ref{prop:coef-construction-local-strength}; their compressed storage
bounds are Corollaries~\ref{cor:coef-storage-finite-degree}
and~\ref{cor:coef-storage-local-strength}.  Proposition~\ref{prop:compressed-response-dp}
provides the compressed dynamic-programming alternative.  Explicit
long-support representation cost is tracked separately by
$\mathcal S_{\rm long}$.

The evaluator bounds used above are summarized here:
\begin{center}
\begin{tabular}{lll}
\hline
Construction & One-time construction & Stored-term bound \\
\hline
Ordered, bounded degree
& Eq.~\eqref{eq:Tcoef-fd} & Eq.~\eqref{eq:Npoly-norm-fd} \\
Ordered, local strength
& Eq.~\eqref{eq:Tcoef-ls} & Eq.~\eqref{eq:Npoly-norm-ls} \\
Compressed dynamic programming
& Eq.~\eqref{eq:compressed-dp-complexity}
& Eq.~\eqref{eq:compressed-dp-state-count} \\
\hline
\end{tabular}
\end{center}

\paragraph{Response-matrix and static-list preprocessing costs.}
The following one-time costs are not included in
Eq.~\eqref{eq:alg2-classical-master-bound}.  Under the uniform bounded-overlap
condition~\eqref{eq:uniform-response-overlap-assumption}, define the dictionary input length
$\mathcal S_{\rm dict}:=\sum_{\alpha\in\mathcal A}
|\operatorname{supp}(\mathcal L_\alpha^\dagger)|$ and the long-Hamiltonian part
$\mathcal S_{\rm long}:=\sum_{C:\,|\operatorname{supp}(C)|>S_D}
|\operatorname{supp}(C)|$.  All dissipative candidates and all short-support
Hamiltonian candidates have constant support, so
$\mathcal S_{\rm dict}=\mathcal{O}(M+\mathcal S_{\rm long})$.  In the unweighted
local-strength case, $\Lambda=s\kappa_0=\mathcal{O}(1)$ and
Lemma~\ref{lem:normalized-term-norms} imply $s=\mathcal{O}(1)$ for a nonempty
dictionary, and hence $\mathcal S_{\rm dict}=\mathcal{O}(M)$.  Bounded dual-interaction-graph degree alone
does not bound the length of a direct-response Hamiltonian string, so the input
cost $\mathcal S_{\rm long}$ must be retained in that case.

The response-row bounds in the proof of
Theorem~\ref{thm:uniform-response-inverse} give
$\operatorname{nnz}(\mathbf G)=\mathcal{O}(M)$ and
$\operatorname{nnz}(\mathbf H)=\mathcal{O}(M)$.  If support-to-candidate incidence lists
or the dual-interaction-graph adjacency lists are available, all nonzero entries of
$\mathbf G$ can be generated in $\mathcal{O}(M)$ arithmetic operations.  If these lists
must first be built from the explicit supports, the expected hashing cost is
$\mathcal{O}(\mathcal S_{\rm dict}+M)=\mathcal{O}(M+\mathcal S_{\rm long})$; naive all-pairs
support testing would instead cost $\mathcal{O}(M^2)$.  Ordering dissipative terms by support size and traversing the
constant-depth, constant-branching extension paths in
$\mathbf G_D^{-1}=\sum_{\ell=0}^{S_D}(-U)^\ell$ constructs $\mathbf H$ in
$\mathcal{O}(M)$ arithmetic operations and $\mathcal{O}(M)$ storage.

There are $M$ static raw-response templates.  Explicitly listing every Pauli
$P$ in an averaged template costs
$\mathcal{O}(MS_D4^{S_D}+\mathcal S_{\rm long})=\mathcal{O}(M+\mathcal S_{\rm long})$; sampling
$P$ on demand avoids storing these $4^{S_D}$ choices.  Replicating only the
row/time labels over $n_t$ measured times costs $\mathcal{O}(Mn_t)$, with $n_t=1$ for
Algorithm~\ref{alg:projected-contraction} and $n_t=r$ for
Algorithm~\ref{alg:cheb}.  Once the empirical raw-response means are supplied
as input, applying the sparse $\mathbf H$ costs
$\mathcal{O}(n_t\operatorname{nnz}(\mathbf H))=\mathcal{O}(Mn_t)$.  Including the one-time sparse
response-matrix preprocessing, the single-time method costs
$\mathcal{O}(M+\mathcal S_{\rm long})+T_{\rm proj}^{\rm cl}(m)$, or
$\mathcal{O}(M)+T_{\rm proj}^{\rm cl}(m)$ when the incidence lists are already available
and all explicitly stored supports have constant length.

\begin{corollary}[Accuracy-dependent classical cost]
\label{cor:alg2-M-dependence}
Assume the uniform bounded-overlap condition
Eq.~\eqref{eq:uniform-response-overlap-assumption} and the unit-cost
symbolic-Pauli model of Proposition~\ref{prop:compressed-response-dp}.  With
the canonical contraction parameters, write
$T_{\rm postmean}^{\rm proj}:=T_{\rm proj}^{\rm cl}(m)$ for the cost after
the empirical means have been formed.  Then
\begin{equation}
    T_{\rm postmean}^{\rm proj}
    =
    \widetilde{\mathcal{O}}\!\left(\frac{M}{\epsilon}\right).
    \label{eq:alg2-postmean-accuracy-cost}
\end{equation}
\end{corollary}

\begin{proof}
Equation~\eqref{eq:explicit-response-entry-bound}, the fact that the diagonal
entries of $\mathbf T$ are at least one, and
$\mathbf H\mathbf G=I_M$ imply $L_C\ge1/4$.  Indeed, for every $\alpha$,
\[
    1
    =
    \left|(\mathbf H\mathbf G)_{\alpha\alpha}\right|
    \le
    \|(\mathbf H)_{\alpha\cdot}\|_1
    \max_j|\mathbf G_{j\alpha}|
    \le
    4\|(\mathbf H)_{\alpha\cdot}\|_1.
\]
For a nonempty normalized dictionary, Eq.~\eqref{eq:unified-locality-scale}
and Lemma~\ref{lem:normalized-term-norms} imply $\Lambda\ge1$.  Hence,
with $\Delta_\star=\max\{1,\lceil\Lambda\rceil\}$ and
$a_\star=\Lambda t_\star$,
\begin{equation}
    5\Delta_\star
    \le
    5(\Lambda+1)
    \le
    4(1+6L_C\Lambda)
    =
    a_\star^{-1}.
    \label{eq:accuracy-dependent-base-comparison}
\end{equation}
The minimal degree prescribed by
Eq.~\eqref{eq:projected-truncation-degree} therefore satisfies
\begin{equation}
    (5\Delta_\star)^m=\mathcal{O}(1/\epsilon).
    \label{eq:accuracy-dependent-degree-growth}
\end{equation}
Equations~\eqref{eq:compressed-dp-complexity}
and~\eqref{eq:alg2-classical-master-bound}, together with
$\operatorname{nnz}(\mathbf H)=\mathcal{O}(M)$ from
Theorem~\ref{thm:uniform-response-inverse} and
$n_{\rm it}=\mathcal{O}(\log(1/\epsilon))$ from
Corollary~\ref{cor:app-projected-iteration-count}, give
\[
    T_{\rm postmean}^{\rm proj}
    =
    \widetilde{\mathcal{O}}(M/\epsilon).
\]
Finally, Eq.~\eqref{eq:app-projected-sample-count} gives
$N_{\rm proj}=\widetilde{\mathcal{O}}(M/\epsilon^2)$, and empirical-mean aggregation
requires $\mathcal{O}(N_{\rm proj})$ classical operations.
\end{proof}

\begin{corollary}[Polylogarithmic-degree regime]
\label{cor:alg2-polylog-M-dependence}
Under the same assumptions, if $m=\mathcal{O}(\log\log M)$, then
\begin{equation}
    T_{\rm proj}^{\rm cl}
    =M\operatorname{polylog}(M),
    \label{eq:alg2-full-M-polylog}
\end{equation}
up to the online iteration factor $\mathcal{O}(\log(1/\epsilon))$ and fixed local
constants.
\end{corollary}

\begin{proof}
Combine Corollary~\ref{cor:compressed-response-dp-polylog} with
Proposition~\ref{prop:algorithm2-classical-cost} and
$n_{\rm it}=\mathcal{O}(\log(1/\epsilon))$.
\end{proof}

For any estimate $\widehat\theta\in\R^M$, choose an
$\ell_\infty$-metric projection onto the physical parameter set,
\begin{equation}
    \widehat\theta^{\rm phys}
    \in
    \operatorname*{arg\,min}_{z\in\Theta_{\rm GKSL}}
    \|z-\widehat\theta\|_\infty.
    \label{eq:optional-physical-projection}
\end{equation}

\begin{corollary}[Stability of the physical projection]
\label{cor:optional-physical-projection}
Let $\theta\in\Theta_{\rm GKSL}$ and suppose
$\|\widehat\theta-\theta\|_\infty\le\epsilon$.  Then every choice in
Eq.~\eqref{eq:optional-physical-projection} is physical and satisfies
\begin{equation}
    \|\widehat\theta^{\rm phys}-\theta\|_\infty
    \le2\epsilon.
    \label{eq:optional-physical-projection-error}
\end{equation}
\end{corollary}

\begin{proof}
The set $\Theta_{\rm GKSL}$ is nonempty and compact, so a minimizer exists.
Since $\theta$ is feasible,
\begin{equation}
    \|\widehat\theta^{\rm phys}-\widehat\theta\|_\infty
    \le
    \|\theta-\widehat\theta\|_\infty
    \le\epsilon.
\end{equation}
The triangle inequality gives
Eq.~\eqref{eq:optional-physical-projection-error}.  Membership in
$\Theta_{\rm GKSL}$ implies $K(\widehat\theta^{\rm phys})\succeq0$, so the
projected generator is a valid Lindbladian.
\end{proof}

Thus a physical estimate with target error $\epsilon$ is obtained by running
the learning procedure with target coordinate error $\epsilon/2$ and then
applying Eq.~\eqref{eq:optional-physical-projection}.

Because $K(z)$ is affine in the real coordinates, the optional projection is a
convex semidefinite program and preserves the candidate zero pattern encoded by
$K(z)$.  Its cost is excluded from the post-processing bounds above.

\end{document}